\documentclass[12pt,a4paper]{article}

\usepackage{amsmath}
\usepackage{amsthm}
\usepackage{amsfonts}
\usepackage{amssymb}
\usepackage{graphicx}
\usepackage{pbox}
\usepackage[section]{placeins}
\usepackage{pifont} 
\usepackage[colorinlistoftodos]{todonotes}
\usepackage{hyperref}

\theoremstyle{plain}
\newtheorem{thm}{Theorem}
\newtheorem{cor}{Corollary}[section]
\newtheorem{propn}{Proposition}[section]
\newtheorem{lem}{Lemma}[section]
\theoremstyle{definition}
\newtheorem{ex}{Example}[section]
\newtheorem{rmk}{Remark}[section]
\newtheorem{defn}{Definition}
\numberwithin{equation}{section}

\DeclareMathOperator{\Hess}{Hess}

\DeclareMathOperator{\crit}{Crit}

\DeclareMathOperator{\Ad}{Ad}
\DeclareMathOperator{\ad}{ad}

\DeclareMathOperator{\Tr}{Tr}

\newcommand{\orb}{\mathcal{O}}

\newcommand{\R}{\mathbf{R}}

\newcommand{\M}{\mathbb{M}}
\newcommand{\inertia}{\mathbb{I}}

\newcommand{\B}{\mathcal{B}}

\newcommand{\Id}{\text{Id}}
\title{Relative Equilibria of Mechanical Systems with Rotational Symmetry}
\author{Philip Arathoon\thanks{University of Michigan, \texttt{philash@umich.edu}}}
\date{June 2024}

\begin{document}
	\maketitle
	\begin{abstract}
		We consider the task of classifying relative equilibria for mechanical systems with rotational symmetry. We divide relative equilibria into two natural groups: a generic class which we call normal, and a non-generic abnormal class. The eigenvalues of the locked inertia tensor descend to shape-space and endow it with the geometric structure of a 3-web with the property that any normal relative equilibrium occurs as a critical point of the potential restricted to a leaf from the web. To demonstrate the utility of this web structure we show how the spherical 3-body problem gives rise to a web of Cayley cubics on the 3-sphere, and use this to fully classify the relative equilibria for the case of equal masses.
	\end{abstract}
	
	
	\tableofcontents
	
	\section*{Background and Outline}
	Finding general solutions of a dynamical system is often far too much to ask. Instead, we redirect our efforts to finding `special' solutions which are more tractable. The equilibria are one such example. There is another example if the system admits a symmetric group action by a Lie group \(G\), namely the \emph{\textbf{relative equilibria (RE)}}. These are solutions which are themselves orbits of one-parameter subgroups of \(G\). 
	
	Perhaps the most famous examples of RE are the central configurations in the planar \(n\)-body problem. These are solutions where the bodies rotate around their centre of mass as if they were a rigid system. They were classified for \(n=3\) by Euler and Lagrange, but for general \(n\) only partial results are known. Even the question of whether there exist finitely many RE for \(n>5\) remains unsolved (Problem~6 in Smale's list \cite{smale_list}).
	
	To motivate the results of this paper it is worth taking a quick look at what is involved in finding RE for the \(n\)-body problem with equal masses. Let \(q\) be the vector of particle positions and \(F\) the vector of forces. The centrifugal force scales linearly with distance, and so to balance the forces in a rotating frame we require
	\[
	F=-\kappa q
	\]
	for some positive \(\kappa\). Force is the negative gradient of the potential \(V\) and the vector \(q\) is half the gradient of the inertia \(\lambda=|q|^2\). Thus, we have a Lagrange-multiplier problem \(2\nabla V=\kappa\nabla\lambda\). Since the inertia and potential are rotationally symmetric, we may as well take the quotient of all configurations with a common centre of mass by \(\mathbf{SO}(2)\) to obtain what is called the shape space. Classifying RE now amounts to finding critical points of the potential restricted to level sets of the inertia in shape space. 
	
	There is a very nice picture of this for the 3-body problem. The level sets of the inertia in shape space are pairs-of-pants, as shown in Figure~\ref{pants} (see \cite[Ch.~14]{sr_book} for a good explanation). One can directly see the RE as the critical points of the potential: the Euler solutions are the three saddle points around the equator, and the two critical points on the top and bottom are the Lagrangian solutions.
	\begin{figure}[h]
		\centering
		\includegraphics{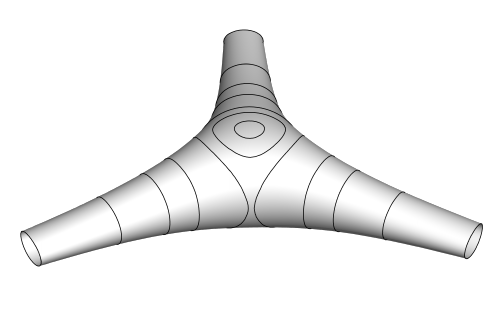}
		\caption{\label{pants}Constant inertia surface in shape space for the planar 3-body problem. The contours are of the potential energy.}
	\end{figure}
	
	The impetus for this work came from asking whether there is a similar nice picture for the spherical 3-body problem, where the bodies are now constrained to move on a sphere. This problem has attracted growing interest in the last few years (we recommend \cite{bitching} for a good review of the literature). Unlike the planar problem, there is no such thing as a centre-of-mass frame, and so one must consider the larger group \(\mathbf{SO}(3)\) of rotations. This means that it no longer makes sense to talk about the inertia in shape space, since the inertia now depends not just on the shape of the bodies but also on the axis of rotation. However, the eigenvalues of the inertia tensor do descend to shape space, and their level sets define a geometric structure which we call a web.
	\begin{defn}
		A \emph{\textbf{k-web}} on a manifold is a collection of regular foliations \(\{\mathcal{F}_1,\dots,\mathcal{F}_k\}\) defined almost everywhere and given locally by the level sets of functions \(\Lambda_1,\dots,\Lambda_k\).
	\end{defn}
	
	We can now state our main result. Let \(Q\) be a configuration space upon which \(\mathbf{SO}(3)\) acts freely and isometrically, and equip it with a potential energy \(V\) which is invariant with respect to the action. 
	
	\begin{thm}\label{introduction_thm}
		The eigenvalues of the locked inertia tensor \(\inertia_q\) descend to shape space \(\mathcal{B}=Q/\mathbf{SO}(3)\) and their level sets endow it with the structure of a 3-web. A point \(x\) in \(\mathcal{B}\) is a normal relative equilibrium if and only if \(2\nabla{V}=\kappa\nabla\lambda_j\) for some multiplicity-1 eigenvalue \(\lambda_j\) and some \(\kappa\ge0\). In particular, \(x\) must be a critical point of \({V}\) restricted to the leaf of constant \(\lambda_j\). The angular momentum of the relative equilibrium is an eigenvector \(L_j\) of \(\inertia_{q}\) with eigenvalue \(\lambda_j\) and magnitude  \(|L_j|^2=\kappa\lambda_j^2\).
	\end{thm}
	In the context of the theorem a normal RE is understood to mean a RE whose angular momentum is not an eigenvector of the inertia tensor for a repeated eigenvalue. The abnormal RE are quite special, since it is precisely where the eigenvalues are repeated that the 3-web becomes singular. Classifying the abnormal RE requires a separate analysis, and are possibly related to the presence of degenerate RE which are not `persistent' in phase space (in the sense described in \cite{jamespersistence}).
	
	To frame how Theorem~\ref{introduction_thm} compares with existing methods of classifying RE some historical comments are in order. Let \(M\) be a symplectic manifold equipped with a Hamiltonian group action \(G\) with momentum map \(\Phi\). In his study of the planar \(n\)-body problem, Smale devises an effective characterisation of the RE as critical points of the energy-momentum map \cite{smale1,smale2}. Thus, if \(H\) is a \(G\)-invariant Hamiltonian on \(M\), the RE are classified by the critical points of \(H\) restricted to a level set \(\Phi^{-1}(\mu)\). Equivalently, one talks of critical points of the augmented Hamiltonian \(H_\xi=H-\langle\Phi,\xi\rangle\).
	
	We can go one step further by passing to the quotient and considering critical points of the reduced Hamiltonian on the symplectic reduced space \(\Phi^{-1}(\mu)/G_\mu\). Thus, for a mechanical system \(M=T^*Q\) we have the task of classifying critical points on a manifold of dimension
	\[
	2\dim Q-\dim G-\dim G_\mu.
	\]
	In particular, for the case of rotational symmetry where \(G=\textbf{SO}(3)\), this will typically involve finding critical points on a manifold of dimension \(2\dim Q-4\). 
	
	In fact, for mechanical systems we can improve upon this still by considering the amended potential
	\[
	V_\mu(q)=V(q)+\frac{1}{2}\langle \mu,\inertia_q^{-1}(\mu)\rangle.
	\]
	The RE of a mechanical system can equivalently be characterised as critical points of \(V_\mu\) \cite{marsden_lectures}. From a computational perspective this is a significant improvement, since we need only compute critical points on a manifold of dimension \(\dim Q\). However, how should one choose \(\mu\) in the amended potential? The magnitude of \(\mu\) clearly makes a difference, however choosing a different \(\Ad_g^*\mu\) on the coadjoint orbit through \(\mu\) shouldn't make much difference, since both the dynamics and momentum map are \(G\)-equivariant.
	
	Theorem~\ref{introduction_thm} can be seen as a way of eliminating this redundancy in the choice of \(\mu\). The amended potential does not descend to shape space. However, by locally describing the configuration space as a principal \(G\)-bundle over shape space, we can introduce a slightly different notion of amended potential \(V_L\) which is defined with respect to the `body' momentum \(L\) (as opposed to the `spatial' momentum \(\mu\)). This function does descend to shape space, and is a crucial step in establishing the theorem.

	
	In summary, this theorem furnishes shape space with additional geometric structure, and provides a less taxing method for classifying the RE. Indeed, RE can now be classified as critical points restricted to a leaf of the web, which will typically have dimension \(\dim Q-4\); an improvement over the amended potential method. In answer to our original question, we do indeed find a nice picture for the spherical 3-body problem, and we invite the reader to skip ahead to Figure~\ref{octych} to behold a 3-web of Cayley cubics in shape space.

	We conclude with a brief outline of the paper. In Section~1 we prove a stronger version of Theorem~\ref{introduction_thm} which applies to any compact group acting freely on configuration space. Section~2 presents a trio of 3-webs arising from physical problems: a rubber ball, the triatomic molecule, and the full-body satellite problem. We show how an understanding of these webs can be used to deduce various existence results for RE. In Section~3 we apply Theorem~\ref{introduction_thm} to obtain a complete and self-contained classification of the RE for the equal-mass spherical 3-body problem. We would like to point out that this classification is not new, having recently been established in a series of preprints by Fujiwara \& P\'{e}rez-Chavela \cite{fuji1,fuji2,fuji3,fuji4,fuji5} and the article \cite{fuji_pub}. In the final section we show how the energy-momentum method of \cite{stability} for assessing stability fits into our formalism. We then apply this to the Eulerian and Lagrangian families of RE for the spherical 3-body; so-called since they generalise the RE of Euler and Lagrange in the planar 3-body problem.

	\section{Main Result}
	Let \(Q\) be configuration space and \(G\) a Lie group which acts on \(Q\). We shall suppose that the orbit-map \(\pi\colon Q\rightarrow \B\) is a smooth principal \(G\)-bundle onto the \emph{\textbf{shape space}} \(\B=Q/G\). We note that this assumption automatically holds if the action is free and the group compact. Additionally, we equip \(Q\) with a \(G\)-invariant metric, and a \(G\)-invariant scalar function \(V\).
	\subsection{Reduction}
	The horizontal subspace \(H_q\) is the orthogonal complement to the tangent space of the \(G\)-orbit passing through \(q\in Q\). The pushforward of \(\pi\) establishes an isomorphism between \(H_q\) and the tangent space to \(x=\pi(q)\). Since the metric on \(Q\) is \(G\)-invariant, the pushforward-metric does not depend on the choice of \(q\) in the fibre of \(\pi\), and hence, we have a well-defined {{reduced metric}} on \(\B\).
	
	Let \(U\) be a chart of \(\B\) and \(\sigma\) a local section of the bundle over \(U\). Consider the decomposition of a tangent vector into its horizontal and vertical parts \(\dot{q}=H+V\). We may identify this with the pair \(\dot{x}=\pi_*(H)\) and \(\dot{g}\), where \(V=\dot{g}\sigma(x)\). This establishes an orthogonal decomposition \(T_UQ=TU\times TG\). As the metric is \(G\)-invariant, the length of a vertical vector \(\dot{g}q\) is the same as \(\omega q\), where \(\omega=g^{-1}\dot{g}\) is the \emph{\textbf{(body) angular velocity}} in the Lie algebra \(\mathfrak{g}\). It follows that 
	\begin{equation}
		|\dot{q}|^2=\langle\M_x(\dot{x}),\dot{x}\rangle+\langle\inertia_{\sigma(x)}(\omega),\omega\rangle
	\end{equation}
	where \(\M_x\colon T_xB\rightarrow T_x^*B\) is the metric tensor for the reduced metric on shape space, and \(\inertia_q\colon\mathfrak{g}\rightarrow\mathfrak{g}^*\) is the symmetric \emph{\textbf{(locked) inertia tensor}} satisfying \(|\omega q|^2=\langle \inertia_q(\omega),\omega\rangle\) for all \(\omega\). 
	
	\begin{rmk}\label{sym2_rep}
		If we choose a different section \(\sigma\mapsto g\sigma\) then the inertia tensor varies according to the representation of \(G\) on \(\text{Sym}^2\mathfrak{g}^*\) given by
		\begin{equation}\label{sym2_rep_eqn}
			\inertia_{gq}=\Ad_g^*\circ~\inertia_{q}\circ \Ad_g^{-1}.
		\end{equation}
	\end{rmk}
	
	We now consider a mechanical system on \(Q\) with Lagrangian \(|\dot{q}|^2/2-V(q)\). The Legendre transform on \(T_UQ=TU\times TG\) sends the tangent vector \((\dot{x},g\omega)\) to \((y,gL)\) where 
	\[
	y=\M_x(\dot{x})
	\]
	is the {{reduced momentum}} in \(T^*_xB\), and 
	\[
	L=\inertia_{\sigma(x)}(\omega)
	\]
	is the \emph{\textbf{(body) angular momentum}} in \(\mathfrak{g}^*\). The Hamiltonian on \(T^*U\times T^*G\) is thus,
	\begin{equation}\label{hamiltonian}
		H(x,y;g,L)=\frac{1}{2}\langle y,\M_x^{-1}(y)\rangle+\frac{1}{2}\langle L,\inertia_{\sigma(x)}^{-1}(L)\rangle+{V}(x).
	\end{equation}
	Observe that this Hamiltonian does not depend on \(g\). Indeed, it descends to the reduced space \(T^*U\times\mathfrak{g}^*\) given by left-translating covectors on \(T^*G\) to \(\mathfrak{g}^*\). We caution that this is not equipped with the standard Poisson structure. This is a consequence of the identification \(T_UQ=TU\times TG\) not being the lift of a diffeomorphism between base spaces.

		The momentum map viewed on \(TQ\) sends \(\dot{q}\) to \(\inertia_{q}(\xi)\) where \(\xi=\Ad_g\omega\) is the \emph{\textbf{(spatial) angular velocity}} and \(q=g\sigma(x)\) (see \cite[Proposition~14.1]{sr_book}). Thus, with the aid of Eq.~\eqref{sym2_rep_eqn} we see that the \emph{\textbf{(spatial) angular momentum}} \(\mu=\Ad_g^*L\) is a conserved quantity. The symplectic reduced space with momentum \(\mu\) is therefore \(T^*U\times\orb\) where \(\orb\) is the coadjoint orbit in \(\mathfrak{g}^*\) through \(\mu\).  

	\subsection{The Amended and Augmented Webs}
	A RE is a solution contained to a group orbit. Hence, RE correspond to fixed points in the reduced space. Equivalently, they are critical points of the reduced Hamiltonian. Let \(\gamma(t)=(x,y,L)\) be a curve  in \(T^*U\times\orb\). Using the fact that \(\M\) and \(\inertia\) are symmetric we have
	\begin{equation}\label{H_first_deriv}
		\frac{d}{dt}H(\gamma(t))=
		\langle\dot{y},\M_x^{-1}y\rangle+\frac{1}{2}\langle y,\dot{\M}_x^{-1}y\rangle+\langle\dot{L},\inertia_x^{-1}L\rangle+\frac{1}{2}\langle L,\dot{\inertia}_x^{-1}L\rangle+\dot{V}.
	\end{equation}
	For \(\gamma(0)=(x_0,y_0,L_0)\) to be a critical point of \(H\) we clearly require \(y_0\) to be zero, and so when we talk of RE we shall only make reference to the \((x_0,L_0)\)-coordinate. We see from Eq.~\eqref{H_first_deriv} that \((x_0,L_0)\) is a RE if and only if \(L_0\) is a critical point of the quadratic form 
	\begin{equation}
		\label{J_defn}
		J_{x_0}(L)=\frac{1}{2}\langle L,\inertia_{x_0}^{-1}L\rangle
	\end{equation}
	restricted to \(\orb\), and \(x_0\) is a critical point of the \emph{\textbf{(reduced) amended potential}}
	\begin{equation}\label{reduced_amended_potential}
		V_{L_0}(x)=V(x)+J_x(L_0).
	\end{equation}
	These two conditions are in a sense coupled. To find a critical point \(x_0\) of \(V_{L_0}\) one must first decide on an appropriate \(L_0\), but finding a critical \(L_0\) of \(J_{x_0}\) itself depends on a choice of \(x_0\). The main idea in the following theorem is to decouple this interdependence by selecting a section \(\sigma\) which `follows' the critical points of \(J_{x_0}\) on a given orbit.

	\begin{defn}
		We will call a relative equilibrium \((x_0,L_0)\) \emph{\textbf{normal}} if \(L_0\) is a non-degenerate critical point of  \(J_{x_0}\) restricted to the coadjoint orbit through \(L_0\). 
	\end{defn}

	\begin{thm}\label{main_thm}

		For every coadjoint orbit \(\orb\) in \(\mathfrak{g}^*\) we may define the \textup{\textbf{amended web}} on shape space defined locally by the level sets of functions \(\Gamma_0,\dots,\Gamma_k\). A point \(x_0\) is a normal relative equilibrium with angular momentum belonging to a scalar multiple of \(\orb\) if and only if \(\nabla V=-\kappa\nabla\Gamma_j\) for some \(j\) and \(\kappa\ge 0\). In particular, \(x_0\) must be a critical point of the potential restricted to a leaf of constant \(\Gamma_j\).
	\end{thm}
	\begin{proof}
		As the coadjoint orbit \(\orb\) is compact it must contain finitely many non-degenerate critical points \(L_0,\dots, L_k\) of the function \(J_{x_0}\). Non-degenerate critical points are stable under perturbations, and so locally in some neighbourhood \(U\) of \(x_0\) we have \(L_j(x)\) with \(L_j(x_0)=L_j\) and which are critical points of  \(J_{x}\) restricted to \(\orb\) for each \(x\in U\).
		
		Introduce the functions \(\Gamma_j(x)=J_x(L_j(x))\). We now make an explicit choice of section \(\sigma\), choosing instead the section \(g(x)\sigma(x)\) where \(g(x)\) satisfies
		\[
		\Ad^*_{g(x)}L_0(x)=L_0.
		\]
		In light of Remark~\ref{sym2_rep} we see that for this choice of section \(
		V_{L_0}=V+\Gamma_0\). Since \(x_0\) must be a critical point of the amended potential at the RE, it follows that 
		\begin{equation}\label{directed_multiplier}
			\nabla V\big|_{x_0}=-\nabla\Gamma_0\big|_{x_0}.
		\end{equation}
	The term on the right hand side is quadratic in \(L_0\), and therefore, if \(\nabla V\) is a negative scalar multiple of \(\nabla\Gamma_0\), then there exists a scalar multiple of \(L_0\) for which Eq.~\eqref{directed_multiplier} holds.
	\end{proof}

We may adapt this construction to define a different web which classifies RE according to their angular velocity \(\omega\) instead of the momentum \(L\). Introduce the quadratic form 
\begin{equation}
	\label{K_defn}
	K_x(\omega)=\frac{1}{2}\langle\inertia_x(\omega),\omega\rangle
\end{equation}
on \(\mathfrak{g}\) and observe that \(J_x(L)=K_x(\omega)\) for \(L=\inertia_x(\omega)\).
\begin{lem}\label{lemma}
	The following are equivalent: \(L=\inertia_x(\omega)\) satisfies
	\begin{equation}\label{coad}
		\ad^*_\omega L=0;
	\end{equation}
	\(L\) is a critical point of \(J_x\) restricted to the coadjoint orbit through \(L\); and, \(\omega\) is a critical point of \(K_x\) restricted to the adjoint orbit through \(\omega\).
\end{lem}
\begin{proof}
	The infinitesimal change in \(J_x(L)\) generated by the tangent vector \(\ad^*_\xi L\) is
	\[
	\frac{1}{2}\langle \ad^*_\xi L,\inertia_x^{-1}L\rangle+\frac{1}{2}\langle L,\inertia_x^{-1}(\ad^*_\xi L)\rangle=\langle \ad^*_\xi L,\omega\rangle=-\langle \ad^*_\omega L,\xi\rangle.
	\]
	This is also the same as the infinitesimal change in \(K_x(\omega)\) generated by the tangent vector \(\ad_\xi\omega\). Therefore, both points are critical if this is zero for all \(\xi\).
\end{proof}
Differentiating the constant  \(\langle\omega_0,\inertia_{x}\inertia^{-1}_{x}L_0\rangle\) reveals that the term \(\langle L,\dot{\inertia}^{-1}_xL\rangle\) in Eq.~\eqref{H_first_deriv} may be replaced with \(-\langle\dot{\inertia}_x(\omega),\omega\rangle\). Together with the previous lemma this gives us an alternative characterisation of RE. The pair \((x_0,L_0)\) is a RE if and only if \(\omega_0=\inertia_{x_0}^{-1}(L_0)\) is a critical point of \(K_{x_0}\) restricted to the adjoint orbit \(\orb\), and if \(x_0\) is a critical point of the \emph{\textbf{(reduced) augmented potential}}
\begin{equation}\label{augmented_pot}
	V_{\omega_0}(x)=V(x)-K_x(\omega_0).
\end{equation}
We will also call a RE normal if \(\omega_0\) is a non-degenerate critical point of \(K_{x_0}\) restricted to \(\orb\). Theorem~\ref{main_thm} may now be modified to
\begin{thm}\label{augmented_thm}
	For every adjoint orbit \(\orb\) in \(\mathfrak{g}\) we may define the \textup{\textbf{augmented web}} on shape space defined locally by the level sets of functions \(\Lambda_0,\dots,\Lambda_k\). A point \(x_0\) is a normal relative equilibrium with angular velocity belonging to a scalar multiple of \(\orb\) if and only if \(\nabla V=\kappa\nabla\Lambda_j\) for some \(j\) and some \(\kappa\ge 0\). In particular, \(x_0\) must be a critical point of the potential restricted to a leaf of constant \(\Lambda_j\).
\end{thm}

	\subsection{The Special Case of Rotational Symmetry}

	For \(G\) compact we may identify \(\mathfrak{g}\) with its dual, and the adjoint representation with the coadjoint representation. The inertia tensor is now identified with a symmetric map \(\inertia_{x}\colon\mathfrak{g}\rightarrow\mathfrak{g}\). Eq.~\eqref{coad} becomes
	\begin{equation}\label{Lie_bracket}
		[\omega,L]=0,
	\end{equation}
	where \([~,~]\) denotes the Lie bracket on \(\mathfrak{g}\). In particular, a necessary condition for \((x,L)\) to be a RE is that \(L=\inertia_x(\omega)\) be contained to the centralizer of \(\omega\).

	\begin{proof}[Proof of Theorem~\ref{introduction_thm}]
		For the Lie algebra \(\mathfrak{so}(3)\) the centralizer of any non-zero \(\omega\) is the line spanned by \(\omega\). Therefore, in the special case where \(\mathfrak{g}=\mathfrak{so}(3)\) the solutions to Eq.~\eqref{Lie_bracket} are the eigenvectors of \(\inertia_x\). Observe that a RE with angular velocity \(\omega\) is normal if and only if the corresponding eigenvalue has multiplicity one

		The orbits of \(\mathfrak{so}(3)\) are spheres of constant \(|\omega|^2\), and so all non-zero orbits are scalar multiples of each other. If we now apply Theorem~\ref{augmented_thm} to this special case then we see that the web is defined by the functions
		\[
		\Lambda_j=K_x(\omega_j)=\frac{1}{2}|\omega_j|^2\lambda_j
		\]
		where \(\omega_j\) is an eigenvector of \(\inertia_{x}\) with eigenvalue \(\lambda_j\).
	\end{proof}
	\begin{rmk}
		For rotational symmetry the augmented and amended webs coincide. One gives the level sets of \(\lambda_j\), and the other gives those of \(\lambda_j^{-1}\). We shall therefore simply refer to `the web' in these cases.
	\end{rmk}
	\begin{ex}[The spherical 2-body problem]\label{2body}
		Consider two non-colinear particles \(q_1\) and \(q_2\) on the unit sphere with masses \(m_1\) and \(m_2\), respectively. Suppose that they are subject to a strictly attractive potential force which only depends on the angle \(\theta\in(0,\pi)\) subtended between them. The orbit-map \(\pi(q_1,q_2)=\theta\) is a principal \(\mathbf{SO}(3)\)-bundle and admits a global section \(\sigma\) sending \(\theta\) to the pair \(q_1=(1,0,0)\) and \(q_2=(\cos\theta,\sin\theta,0)\). The locked inertia tensor for this pair of points is 
		\begin{equation}
			\inertia_\theta=\begin{pmatrix}
				m_2	\sin^2\theta & -\frac{m_2}{2}\sin 2\theta& 0\\
				-\frac{m_2}{2}\sin 2\theta & m_1+m_2\cos^2\theta & 0\\
				0 & 0 & m_1+m_2
			\end{pmatrix}
		\end{equation}
		with characteristic polynomial 
		\begin{equation}
			(m_1+m_2-t)\left(2t-m_1-m_2-\sqrt{D_\theta}\right)\left(2t-m_1-m_2+\sqrt{D_\theta}\right),
		\end{equation}
		where \(D_\theta=m_1^2+2m_1m_2\cos 2\theta+m_2^2\). The eigenvalue \(\lambda_0=m_1+m_2\) is constant in \(\theta\), and therefore defines a trivial leaf of the web equal to the whole of shape space. The two remaining eigenvalues
		\begin{equation}
				\lambda_{\pm}=\frac{m_1+m_2\pm\sqrt{D_\theta}}{2}
		\end{equation}
		define leaves of the web which are discrete points. The derivatives \(d\lambda_{\pm}/d\theta=\mp m_1m_2\sin 2\theta/\sqrt{D_\theta}\) have opposite signs, and hence, for a strictly attractive force with \(dV/d\theta>0\), there exists a single normal RE (up to time-reversal symmetry) for every \(\theta<\pi/2\) with inertia \(\lambda_-\), and every \(\theta>\pi/2\) with inertia \(\lambda_+\).
	
		To handle abnormal RE we must first identify the \(x_0\) in shape space for which the inertia tensor admits repeated eigenvalues. Once we have done so, we must find when \(x_0\) is a critical point of \(V_{\omega}\) as \(\omega\) ranges over the repeated eigenspace. In this example, the inertia \(\inertia_\theta\) admits repeated eigenvalues when \(m_1=m_2=m\) and for \(\theta=\pi/2\). In this case, the eigenspace for the repeated eigenvalue is the plane \(\omega=(u,v,0)\) with \(u,v\in\R\). A computation reveals
		\begin{equation}
			\frac{1}{2}\frac{d}{d\theta}\Big|_{\theta=\pi/2}\langle \inertia_{\theta}(\omega),\omega\rangle=muv
		\end{equation}
		from which we conclude that an entire family of abnormal RE exist for the precise case of equal masses and \(\theta=\pi/2\), provided \(muv=|V'(\pi/2)|\). This completes the classification of RE for the spherical 2-body problem \cite{jamesborisov}.
	\end{ex}

	\section{Examples of Webs}
	The most interesting examples of 3-webs are those which we can visualise in 3-dimensional space. We should therefore like to gather examples of 6-dimensional configuration spaces \(Q\) upon which \(\mathbf{SO}(3)\) acts freely. Below we enumerate four examples together with a physical motivation for each.
	\[
	\renewcommand{\arraystretch}{1.2}
	\begin{array}{l|l}
		Q &  \text{Physical example}\\\hline
		\text{Sym}^2\R^3 & \text{a rubber ball}\\
		\R^3\times\R^3 & \text{a triatomic molecule}\\
		\mathbf{SE}(3) &  \text{a body in a central force field }\\
		\mathbf{S}^2\times 	\mathbf{S}^2\times 	\mathbf{S}^2 &\text{the spherical $3$-body problem}
	\end{array}
	\]
	In this section we shall exhibit the 3-webs in the first three examples, presenting them in increasing order of complexity. For the rubber ball the web is given by linear planes, for the triatomic molecule they are planes and quadratic cones, and for the orbital satellite they are an intriguing family of quartic surfaces. 
	
	For more general symmetry groups there are more solutions to Eq.~\ref{Lie_bracket} than just the eigenvectors of \(\inertia_x\). For a fixed \(x\) this equation is quadratic in \(\omega\in\mathfrak{g}\) and the solutions depend very delicately upon the nature of the Lie algebra and the inertia tensor. Furthermore, for larger groups there is a greater variety of co/adjoint orbits. Consequently, there exists a family of augmented/amended webs parametrised by the equivalence classes of co/adjoint orbits, modulo scaling. 
	
	An example of a mechanical system with a larger symmetry group is the Dirichlet system, whose RE were famously classified by Riemann and are known as the Riemannian ellipsoids \cite{riemann}.
	\[
	\renewcommand{\arraystretch}{1.2}
	\begin{array}{l|l}
		Q &  \text{Physical example}\\\hline
		\textbf{SL}(3)& \text{an incompressible liquid drop}
	\end{array}
	\]
	This \(Q\) admits a symmetric group action by \(\mathbf{SO}(3)\times\mathbf{SO}(3)\), whose adjoint orbits are products of spheres \(\textbf{S}^2_{\rho_1}\times\textbf{S}^2_{\rho_2}\) with radii \(\rho_{1,2}\). We conclude the section by demonstrating how this defines a family of augmented 6-webs on shape space parametrised by the ratio \(\rho_1:\rho_2\).
	
	\subsection{The Rubber Ball}
	 Suppose we have a rubber ball which may be deformed into the shape of an ellipsoid. We shall model the deformation of the ball by supposing that at time \(t\) the particle initially at \(x_0\) is now at \(S(t)x_0\), where \(S\) is a symmetric matrix. The kinetic energy is given by
	\begin{equation*}
		\frac{1}{2}\int |\dot{S}x_0|^2~dx_0=\frac{1}{2}\Tr (\dot{S}^T\dot{S}),
	\end{equation*}
	where we have assumed that \(\int x_0x_0^Tdx_0=\Id\). This defines a metric on the configuration space \(\text{Sym}^2\R^3\) which is invariant with respect to conjugation by \(\mathbf{O}(3)\). The orbit map \(\pi\colon\text{Sym}^2\R^3\rightarrow\R^3_{\ge 0}/\mathbf{S}_3\) for this action sends \(S\) to its eigenvalues \((x_1,x_2,x_3)\) modulo permutations. 
	\begin{rmk}\label{glib}
		We note that this shape space is not a smooth manifold, nor is the action of \(\mathbf{O}(3)\) on \(\text{Sym}^2\R^3\) free. However, we can simply sidestep these technicalities by working locally in a region of shape space where \(x_1<x_2<x_3\). 
	\end{rmk}
	
	We choose a local section which sends the eigenvalues to the diagonal matrix \(x=\text{diag}(x_1,x_2,x_3)\). The kinetic energy generated by \(\omega\in\mathfrak{so}(3)\) acting on the configuration \(x\) is thus
	\[
	\frac{1}{2}\Tr\left((\omega x-x\omega)^T(\omega x-x\omega)\right)=\frac{1}{2}\langle\omega,\inertia_x(\omega)\rangle.
	\]
	If we identify \(\mathfrak{so}(3)\) with \(\R^3\) in the standard way then \(\inertia_x\) becomes the diagonal matrix with entries \((x_1-x_2)^2, (x_1-x_3)^2, (x_2-x_3)^2\). We therefore have
	\begin{propn}
		The 3-web on shape space \(\R^3_{\ge 0}/\mathbf{S}_3\)  is given by the planes of constant \(x_i-x_j\). The inertia tensor has twice-repeated eigenvalues precisely on the planes \(2x_k=x_l+x_m\).
	\end{propn}

	\subsection{The Triatomic Molecule}
	Consider \(n\) particles located at \(q_1,\dots, q_n\) in \(\R^3\) with masses \(m_1,\dots, m_n\). The kinetic energy generated by the angular velocity \(\omega\in\mathfrak{so}(3)\cong\R^3\) is
	\begin{align*}
		\frac{1}{2}\sum m_j|\omega\times q_j|^2
		&=\frac{1}{2}\sum m_j\left(|\omega|^2|q_j|^2-\langle\omega, q_j\rangle^2\right)\\
		&=\frac{1}{2}\omega^T\left[\left(\sum m_j|q_j|^2\right)\Id-\sum m_jq_jq_j^T\right]\omega.
	\end{align*}
	Let \(Q\) denote the matrix whose \(j\)\textsuperscript{th}-column is \(\sqrt{m_j}q_j\). The kinetic energy is then  \(\frac{1}{2}\langle \omega,\inertia_{Q}\omega\rangle\) where
	\begin{equation}\label{nboday_inertia}
		\inertia_{Q}=|{Q}|^2\Id-{Q}{Q}^T
	\end{equation}
	is the inertia tensor and \(|{Q}|^2=\Tr({Q}^T{Q})\). 
	
	We shall now limit our attention to the case where \(n=3\)  and suppose that we are in a centre-of-mass frame within which \(\sum m_jq_j=0\). The configuration is therefore entirely determined by \(q_1\) and \(q_2\) alone. The shape space may be identified with the solid cone \(\mathcal{C}\) in \(\R^3\) defined by
	\[
	x_1x_2-x_3^2\ge 0,\quad\text{and}\quad x_1,x_2\ge0,
	\]
	where \(x_1=|q_1|^2\), \(x_2=|q_2|^2\), and \(x_3=\langle q_1, q_2\rangle\). We choose an orbit-map \(\pi\) sending \((q_1, q_2)\) to
	\begin{equation}
		\pi({Q})=|{Q}|^2\textup{\Id}-{Q}^T{Q},
	\end{equation}
	which we note to be a linear expression in \(x_1, x_2, x_3\).
	
	\begin{propn}
		Let \(\{\hat{\mu},\hat{\xi}, \hat{\eta}\}\) be an orthonormal basis of \(\R^3\) where \(\mu=(\sqrt{m_1},\sqrt{m_2},\sqrt{m_3})^T\), and define the symmetric matrix
		\begin{equation}
			{S}=2\hat{\mu}\hat{\mu}^T+\hat{\xi}\hat{\xi}^T+\hat{\eta}\hat{\eta}^T.
		\end{equation}
		This matrix belongs to the interior of \(\mathcal{C}\), and the inertia tensor has twice-repeated eigenvalues precisely along the line spanned by \(S\) and on the boundary \(\partial\mathcal{C}\). The 3-web in \(\mathcal{C}\) consists of the translations along this line of: the plane \(\Tr\pi(Q)=0\), the cone \(\partial\mathcal{C}\), and the reversed cone \(-\partial\mathcal{C}\). 
	\end{propn}
	
	\begin{proof}
		Observe that the characteristic polynomials of \(\inertia_{Q}\) and \(\pi(Q)\) coincide. Therefore, it suffices to study the eigenvalues of \(\pi(Q)\). 
		
		The centre-of-mass condition is equivalent to \(Q\mu=0\). Therefore, \(\pi(Q)\) always has \(\mu\) as an eigenvector with eigenvalue \(|Q|^2\). The matrix \(S\) is the image under \(\pi\) of \(\hat{\xi}\hat{\xi}^T+\hat{\eta}\hat{\eta}^T\), and has a repeated eigenvalue in the plane orthogonal to \(\mu\). Consequently, we may translate any leaf of the web by some multiple of \(S\) to obtain a leaf with a constant zero eigenvalue. It therefore suffices to consider the leaves where an eigenvalue is zero: the subset with \(|Q|^2=0\) is the plane \(\Tr\pi(Q)=0\), and if any other eigenvalue is zero then the particles are colinear, and hence \(x_1x_2=x_3^2\).
	\end{proof}
	A picture of the components of this 3-web is given in Figure~\ref{triatomic_web}. An understanding of the geometry of the 3-web can be very useful in determining the existence of RE. For instance, suppose that a given level set of the potential is a bounded subset in shape space. Consider the effect of translating the leaves of the 3-web along the line spanned by \(S\). The planes must intersect the level set tangentially at least twice. So too must the forward  and reverse cones. We therefore have a pair of critical points of the reduced potential for each constant eigenvalue. It can be shown that for at least three of these critical points the gradient of the potential is oriented in the direction of increasing eigenvalue. Theorem~\ref{introduction_thm} applied to this idea then gives a result of Montaldi \& Roberts \cite{jamesmolecule} concerning the bifurcation of an equilibrium into 3 families of RE.
	
	\begin{figure}
		\centering
		\includegraphics[scale=0.6]{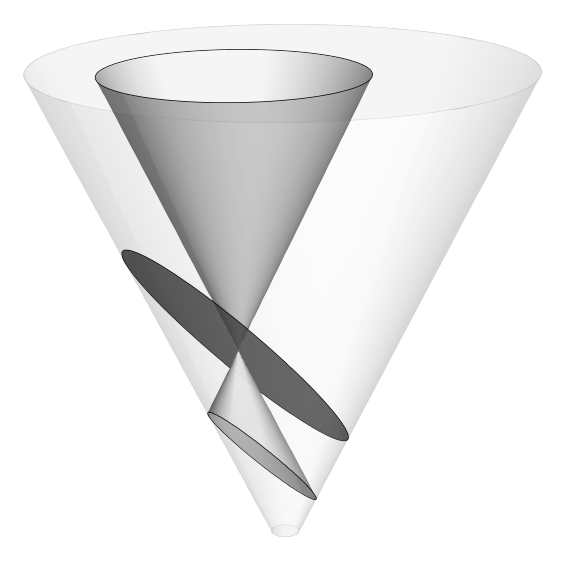}
		\caption{\label{triatomic_web}The 3-web in shape space for the triatomic molecule.}
	\end{figure}
	
	\begin{propn}
		Consider a triatomic molecule in a stable non-degenerate equilibrium configuration \(Q_0\). For any \(\epsilon>0\) sufficiently small, there exist at least 3 distinct normal relative equilibria at \(\pi(Q_j)\) with \(|\pi(Q_0)-\pi(Q_j)|<\epsilon\) for \(j=1,2,3\).
	\end{propn}
	
	\begin{rmk}
		More generally suppose that \(x_0\) is a stable non-degenerate equilibrium in shape space. For a given coadjoint orbit suppose the \(k\)-web is regular at \(x_0\). Then in a sufficiently small neighbourhood of \(x_0\) the equilibrium bifurcates into a family of at least \(k\) normal RE.
	\end{rmk}
	\subsection{The Full-Body Satellite Problem}
	Consider the motion of a rigid body in \(\R^3\) subject to a central force directed through the origin. The configuration space may be taken to be \(\mathbf{SE}(3)=\mathbf{SO}(3)\ltimes\R^3\), whereby a configuration of the body is identified with the Euclidean motion \((a,d)\) which sends the fixed frame \(\mathcal{F}\) in space to the body frame \(a\mathcal{F}+d\). 
	
	Consider the action of \(\mathbf{SO}(3)\) on the left given by \(g(a,d)=(ga,gd)\). An orbit-map for this action is 
	\[
	\pi\colon \mathbf{SE}(3)\longrightarrow\R^3;\quad(a,d)\longmapsto x=-a^{-1}d.
	\]
	The vector \(x\) is the origin of the space-frame viewed from within the body frame. We choose \(\sigma\colon x\mapsto (\Id,-x)\) to be a section for this principal bundle and consider the kinetic energy generated by \(\omega\in\mathfrak{so}(3)\) acting on the configuration \(\sigma(x)\)
	\[
	\frac{1}{2}\int|\omega z-\omega x|^2dz=\frac{1}{2}\int\left(|\omega z|^2-2\langle \omega z,\omega x\rangle+|\omega x|^2\right) dz.
	\]
	The integral is taken over the points \(z\) belonging to the body in the body frame. If we suppose that the origin of the body frame coincides with the centre of mass of the body, then \(\int z dz=0\) and the kinetic energy simplifies to
	\[
	\frac{1}{2}\langle \inertia_0(\omega),\omega\rangle+\frac{1}{2}|\omega x|^2
	\]
	where \(\inertia_0\) is the constant inertia tensor for the body in the body frame, and where we are supposing for simplicity that the body has unit mass. By once again identifying \(\omega\) with an element in \(\R^3\) we can write the kinetic energy as \(\frac{1}{2}\langle \omega,\inertia_{x}(\omega)\rangle\) where the symmetric matrix
	\begin{equation}\label{fullbody_inertia}
		\inertia_{x}=\inertia_0+|x|^2\Id-xx^T
	\end{equation}
	is the inertia tensor. We may suppose that the body frame is chosen so that \(\inertia_0\) is diagonal, with distinct entries \(I_1< I_2< I_3\), the principal moments of inertia  of the body.
	
	\begin{propn}
		The 3-web on shape space is given by the surfaces 
		\begin{equation}\label{char_poly_fullbody}
			1=\sum_{j=1,2,3}\frac{x_j^2}{(I_j-\lambda)+|x|^2}
		\end{equation}
		for \(\lambda\ge I_1\) a given eigenvalue of \(\inertia_{x}\). The subset of shape space in which the inertia tensor admits repeated eigenvalues are the curves
		\begin{equation}\label{equi_eigencurves_fullbody}
			x_k^2(I_l-I_m)+x_l^2(I_k-I_m)=(I_k-I_m)(I_l-I_m),\quad x_m=0,
		\end{equation}
		for \(k,l,m=1,2,3\) all distinct.
	\end{propn}
	\begin{proof}
		One may verify that the characteristic polynomial of \(\inertia_x\) is the numerator obtained from Eq.~\eqref{char_poly_fullbody} after multiplying up by the denominators. 
		
		Suppose \(\inertia_x\) has a repeated eigenvalue \(\lambda\). The eigenspace must contain an eigenvector \(\omega\) orthogonal to \(x\), and so \(\inertia_x(\omega)=\inertia_0\omega+|x|^2\omega=\lambda\omega\), which implies that \(\omega\) is along a principal direction, let's say \(e_m\), and that \(\lambda=I_m+|x|^2\). Now take a second independent eigenvector \(\nu\), which we may suppose is orthogonal to \(\omega\). Then
		\begin{equation}\label{nu_eigenvector}
			\inertia_x(\nu)=\inertia_0\nu+|x|^2\nu-x\langle x,\nu\rangle=\lambda\nu=(I_m+|x|^2)\nu,
		\end{equation}
		which implies \(\inertia_0\nu-I_m\nu=x\langle x,\nu\rangle\). It follows that \(\nu\) is proportional to \(D^{-1}x\), where \(D=\text{diag}(I_k-I_m,I_l-I_m)\) with respect to the basis \(\{e_k,e_l\}\). Substituting this \(\nu\) back into Eq.~\eqref{nu_eigenvector} yields Eq.~\eqref{equi_eigencurves_fullbody}.
	\end{proof}
	%
	
	The surfaces defined in Eq.~\eqref{char_poly_fullbody} are quartic in \(x\) and come in three families depending on which of the intervals \((I_1,I_2)\), \((I_2,I_3)\), or \((I_3,\infty)\) the eigenvalue \(\lambda\) belongs to (and three singular leaves for when \(\lambda=I_1,I_2,I_3)\). The regular leaves are displayed in Figure~\ref{fullbody_web}.
	\begin{itemize}
		\item[\ding{226}]The leaves are empty for \(\lambda<I_1\). Indeed, the inertia in space can never be less than the smallest moment of inertia of the body itself. 
		\item[\ding{226}] For \(I_1<\lambda<I_2\) the leaf is a topological cylinder aligned with the \(x_1\)-axis.
		\item[\ding{226}] For \(I_2<\lambda<I_3\) the leaf consists of two sheets joined together with a topological sphere at 4 singular points in the plane \(x_1=0\).
		\item[\ding{226}]For \(\lambda>I_3\) the leaf is two concentric topological spheres joined together at 4 nodal singularities in the plane \(x_2=0\). Figure~\ref{fullbody_web} shows two cutaways of this surface to highlight this property.
	\end{itemize}

	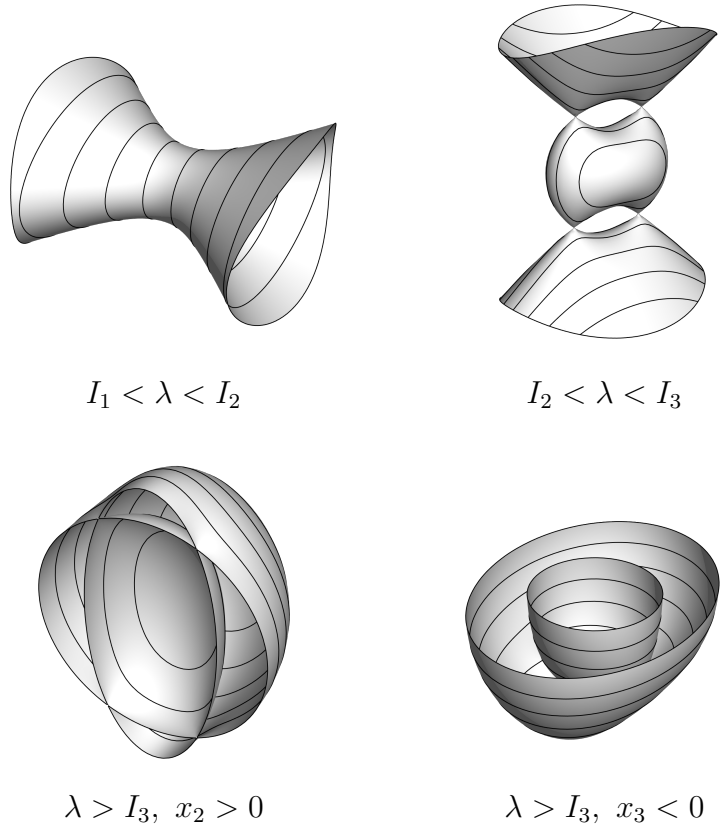
\begin{figure}[h]
		\centering
\begin{tikzpicture}
	\draw (-3, 3) node[inner sep=0] {\includegraphics[scale=0.9]{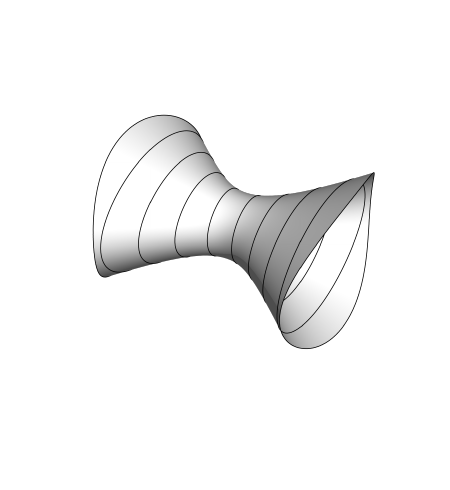}};
	\draw (-3,0.5) node {$I_1<\lambda<I_2$};
	
	\draw (3,3) node[inner sep=0] {\includegraphics[scale=0.9]{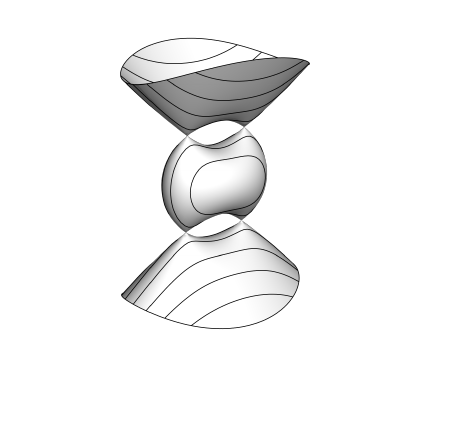}};
	\draw (2.8,0.5) node {$I_2<\lambda<I_3$};
	
	\draw (-2.5,-2) node[inner sep=0]
	{\includegraphics[scale=0.9]{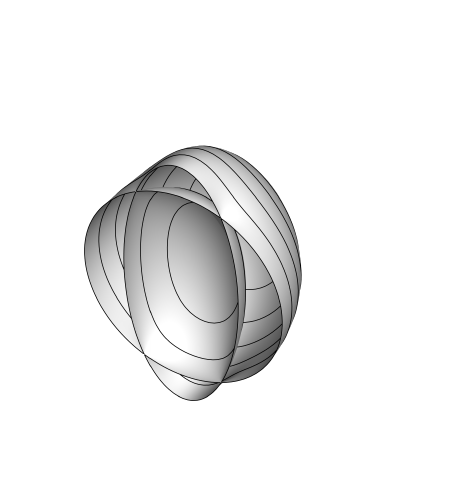}};
	\draw (-3,-5) node {$\lambda>I_3,~x_2>0$};
	
	\draw (2.5,-3.5) node[inner sep=0] {\includegraphics[scale=0.9]{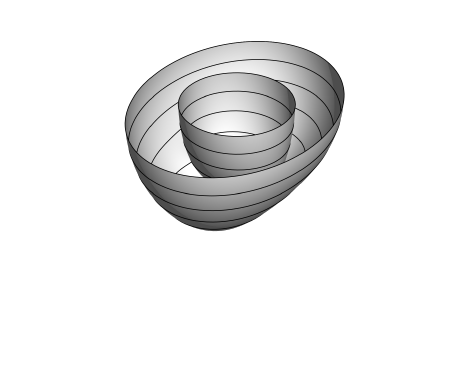}};
	\draw (2.8,-5) node {$\lambda>I_3,~x_3<0$};
	
\end{tikzpicture}
		\caption{\label{fullbody_web}The 3-web in shape space for the full-body satellite problem.}
	\end{figure}
	
	\begin{rmk}
		The curves in Eq.~\eqref{equi_eigencurves_fullbody} are empty if \(I_m\) is the largest moment of inertia. Therefore, there are 3 curves of repeated eigenvalues: an ellipse in the plane \(x_1=0\), and two hyperbolae in the plane \(x_2=0\). Notice how these curves match the singular points of the web. Indeed, singular points of a web are indicative of repeated eigenvalues.
	\end{rmk}
	

	A \emph{\textbf{great-circle RE}} is a RE for which the centre of mass of the body remains in the plane of rotation \cite{krish}. For our choice of section \(\sigma\) this is equivalent to \(\omega\) being orthogonal to \(x\). From Eq.~\ref{fullbody_inertia} this implies that \(\omega\) is parallel with a principal direction. In fact, we have
	
	\begin{lem}\label{greatcircle_lemma}
		A relative equilibrium at \(x\in\R^3\) is a great-circle motion if and only if \(x\) belongs to a principal plane, say orthogonal to \(e_m\), and the corresponding angular velocity \(\omega\) is parallel to \(e_m\) with inertia \(\lambda=I_m+|x|^2\).
	\end{lem}

	Consider the orbit of a body sufficiently far away from the origin so that its orbital inertia is larger than any moment-of-inertia of the body itself. From topological considerations of the leaf alone we are able to deduce the following result of \cite{wang}.
	
	\begin{propn}
		For any sufficiently large \(\lambda>I_3\) there exist at least two normal relative equilibria for a rigid body in an attractive central-force with potential \(V\). If the body has a plane of symmetry, then it has at least 2 great-circle relative equilibria. If the body is symmetric with respect to 2 planes of symmetry, then it has at least 4 great-circle relative equilibria, and if the body has 3 planes of symmetry then it has at least 12 great-circle relative equilibria (all solutions defined up to time-reversal).
	\end{propn}
	\begin{proof}
		We suppose the body does not intersect the origin in space. Therefore, both V and \(\lambda\) increase in the radial direction, and hence, if \(\nabla V\) and \(\nabla \lambda\) are parallel, then they must be in the same direction. It follows that all critical points of the potential and inertia are normal RE.
		
		On the leaf \(\lambda>I_3\) there must exist a minimum of \(V\) on the inner sphere, and a maximum on the outer sphere (neither of which can coincide with the 4 singular conical points of the leaf). 
		
		The intersection of the leaf with the principal plane orthogonal to \(e_m\) is the union of an ellipse and the circle \(|x|^2=\lambda- I_m\). The potential must have at least two critical points on such a circle, and hence, from symmetry considerations and Lemma~\ref{greatcircle_lemma} the result follows. 
	\end{proof}

\subsection{The Riemannian Ellipsoids}
Consider a spherical droplet of incompressible fluid at time \(t=0\). In the linear approximation the configuration of the droplet at time \(t\) is determined by a \(Q(t)\in\mathbf{SL}(3)\), whereby the fluid particle initially at \(x_0\) is now at \(Q(t)x_0\). If the fluid is homogeneous and has unit mass, then the kinetic energy is \(\Tr(\dot{Q}^T\dot{Q})\), which we observe to be invariant under both left and right multiplication by \(\mathbf{O}(3)\).

The singular values \(\{x_1,x_2,x_3\}\) of \(Q\) are invariant under this action, and so we may identify shape space with \(\mathcal{B}=\{x_1,x_2,x_3\ge 0~|~x_1x_2x_3=1\}\). Strictly speaking, shape space is actually the quotient of this set by permutations, however we won't come to any harm if we proceed by working locally in a region where the singular values are distinct, as in Remark~\ref{glib}.

We take a section over shape space by sending the singular values to the matrix \(\text{diag}(x_1,x_2,x_3)\). If we identify elements  of \(\mathfrak{so}(3)\) with vectors in \(\mathbf{R}^3\) then for this choice of section the inertia tensor is the block matrix
\begin{equation}
	\inertia_x=\begin{pmatrix}
		A & B\\B & A
	\end{pmatrix},~\text{for}~A=\text{diag}(a_1,a_2,a_3),~\text{and}~B=\text{diag}(b_1,b_2,b_3),
\end{equation}
where \(a_k=x_l^2+x_m^2\) and \(b_k=-2x_lx_m\) for distinct \(k,l,m\). The vectors \(\omega,\xi\in\mathfrak{so}(3)\) are the angular velocity and vorticity of the fluid. The terminology in the following lemma is borrowed from \cite{fasso}.

\begin{lem}\label{ellipsoid_lemma}
	For distinct \(x_1,x_2,x_3,\) the solutions to \([(\omega,\xi),\inertia_x(\omega,\xi)]=0\) come in two types:
	\begin{enumerate}
		\item The vectors \(\omega\) and \(\xi\) are each parallel to a principal direction \(e_k\) and have arbitrary length. These solutions are said to be of \textup{\textbf{Type}}~\textbf{\textup{S\textsubscript{k}}}.
		\item For any \(\omega\) orthogonal to \(e_k\) but not parallel with \(e_l\) or \(e_m\), there exist exactly two solutions for \(\xi\) if
		\[
		D=(2x_k-x_l-x_m)(2x_k+x_l-x_m)(2x_k-x_l+x_m)(2x_k+x_l+x_m)
		\]
		is strictly positive, and precisely one if it is equal to zero. These solutions are said to be of \textup{\textbf{Type}}~\textbf{\textup{R\textsubscript{k}}}. If \(D\) is negative then there are no solutions of this type.
	\end{enumerate}
\end{lem}
\begin{proof}
	The Lie bracket \([(\omega,\xi),(L,\Omega)]\) is \((\omega\times L,\xi\times\Omega)\), and so we require
	\begin{align}\label{s_t_system}
		\begin{split}
			A\omega+B\xi&=s\omega\\
			B\omega+A\xi&=t\xi
		\end{split}
	\end{align}
	for some \(s,t\in\R\). We can solve for \(\xi\) in the first equation and then substitute this into the second to find that \(C\omega=0\), where \( C=B^2-A^2+(s+t)A-(st)\Id\). This matrix is diagonal, and so its kernel is either: a principle direction, giving solutions of Type~$\text{S}$; or, a principal plane. In the latter, for \(e_l\) and \(e_m\) to both belong to the kernel of \(C\) we require \((s,t)=(u_\pm,u_\mp)\), where
	\begin{equation}
		u_\pm=\frac{1}{2}\left(x_l^2+x_m^2-2x_k^2\pm \sqrt{D}\right).
	\end{equation}
\end{proof}

In general we expect \(\inertia_x\) to be a less complicated function in \(x\) than its inverse. For this reason, it will be more convenient to present the augmented web. Thanks to the previous lemma, for a given adjoint orbit \(\mathbf{S}^2_{\rho_1}\times\mathbf{S}^2_{\rho_2}\) the web is given by the level sets of 
\[
\Lambda(x)=\langle\inertia_x(\omega,\xi),(\omega,\xi)\rangle=s\rho_1^2+t\rho_2^2
\]
where \(s\) and \(t\) are as in Eq.~\eqref{s_t_system}. 

There is a slight caveat. Solutions of Type R\textsubscript{k} only exist on the given orbit in a subset of shape space. To see why, consider the map \(\omega\mapsto\xi=B^{-1}(u_\pm\omega-A\omega)\) obtained from solving for \(\xi\) in Eq.~\eqref{s_t_system}. This sends the circle in the plane orthogonal to \(e_k\) with radius \(|\omega|=\rho_1\) into an ellipse. For there to exist a solution on the orbit this ellipse must intersect the circle \(|\xi|=\rho_2\). This defines a region \(\mathcal{R}^\pm_{k}\subset\mathcal{B}\).  
\begin{propn}\label{ellisoid_web_propn}
	The augmented web on \(\mathcal{B}\)  corresponding to the adjoint orbit \(\mathbf{S}^2_{\rho_1}\times\mathbf{S}^2_{\rho_2}\) consists of two components:
	\begin{enumerate}
		\item  The web of Type~S given by the level sets of 
		\begin{equation}\label{web_Stype}
			\Lambda(x)=(\rho_1^2+\rho_2^2)(x_l^2+x_m^2)\pm 4\rho_1\rho_2 x_lx_m
		\end{equation}
		defined everywhere. 
		\item The web of Type~R given by the level sets of 
		\begin{equation}
			\Lambda(x)=\frac{1}{2}(\rho_1^2+\rho_2^2)(x_l^2+x_m^2-2x_k^2)\pm\frac{\sqrt{D}}{2}(\rho_1^2-\rho_2^2)
		\end{equation}
		defined in the region \(\mathcal{R}^\pm_{k}\).
	\end{enumerate}
	These equations are defined for each choice of distinct \(k,l,m\), and for each \(\pm\)-sign.
\end{propn}
An example of the full 6-web of Type~S is shown on the left in Figure~\ref{spider_web} for \(\rho_1/\rho_2=3\). In this figure we have identified \(\mathcal{B}\) with the plane \(z_1+z_2+z_3=0\) by taking \(z_k=\log x_k\). Apparent in this figure is the \(\mathbf{S}_3\)-symmetry which arises from permuting the singular values of \(Q\). Shown on the right is a component of the web of Type R, together with the shaded region inside which it is defined.

\begin{figure}
	\centering
	\includegraphics[scale=0.65]{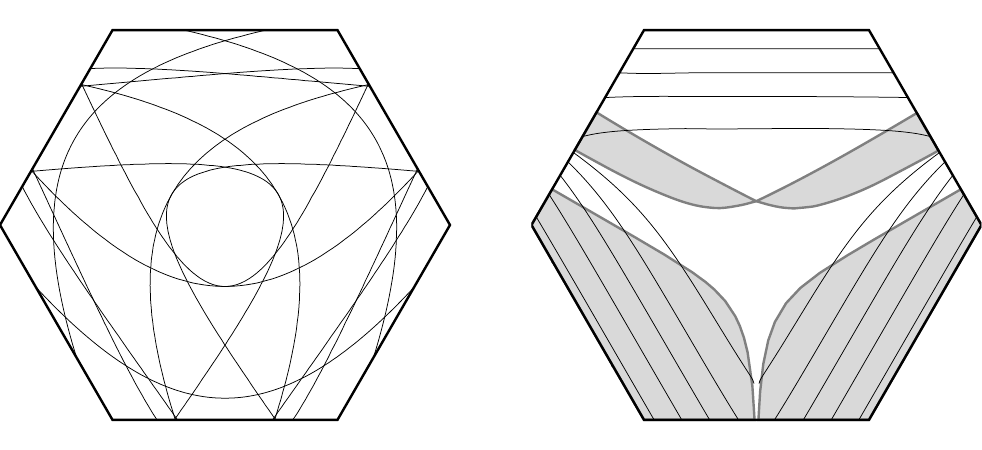}
	\caption{\label{spider_web}Webs of Type S and Type R}
\end{figure}

\begin{rmk}
	Dedekind famously observed \cite{dedkind} that if \(Q(t)\) was a RE then so was the transpose \(Q^T(t)\). Taking the transpose of a RE swaps the angular velocity and vorticity vectors \(\omega\leftrightarrow\xi\). We can see this symmetry at the level of webs by observing that the web in Proposition~\ref{ellisoid_web_propn} is invariant with respect to \(\rho_1\leftrightarrow\rho_2\).
\end{rmk}

	\section{The Spherical 3-Body Problem}
	Consider 3 particles \(q_1, q_2, q_3\) on the unit sphere and let \(Q\) be the matrix whose \(j\)\textsuperscript{th}-column is \(q_j\). An orbit map for the action of \(\textbf{O}(3)\) on \(\textbf{S}^2\times\textbf{S}^2\times\textbf{S}^2\) is given by sending \((q_1,q_2,q_3)\) to 
	\begin{equation}
		\pi(Q)=3\Id-Q^TQ=\begin{pmatrix}
			2 & -x_1 & -x_2\\
			-x_1 & 2 & -x_3\\
			-x_2& -x_3 & 2
		\end{pmatrix}
	\end{equation}
	where 
	\[
	x_1=\cos\theta_{12},\quad x_2=\cos\theta_{13},\quad x_3=\cos\theta_{23},
	\]
	and \(\theta_{ij}\) is the angle subtended between \(q_i\) and \(q_j\). The \(x_j\) are bounded between \(-1\) and \(1\) and satisfy the the inequality
	\[
	\det Q^TQ=|\langle q_1,q_2\times q_3\rangle|^2\ge 0.
	\]
	The shape space is therefore the \emph{\textbf{(curvy) tetrahedron}} \(\mathcal{T}\subset\mathbf{R}^3\) given by the cubic inequality
	\begin{equation}\label{curvytetrahedron_defn}
		1+2x_1x_2x_3-x_1^2-x_2^2-x_3^2\ge 0.
	\end{equation}
	It is helpful to see which parts of the tetrahedron correspond to given configurations of the particles.
	\begin{itemize}
		\item[\ding{226}] The vertices of \(\mathcal{T}\) are the colinear configurations (which we should technically exclude since here the action is not free). The top vertex \(P\) at \(x_1=x_2=x_3=1\) is the triple collision, and the other vertices are pairs of binary collisions with the third particle antipodal to the pair. 
		\item[\ding{226}] The three edges from \(P\) correspond to binary collisions, and the other three edges to antipodal pairs.
		\item[\ding{226}] The boundary \(\partial\mathcal{T}\) consists of the coplanar configurations. The faces \(F_k\) with common vertex \(P\) are those configurations where the particles lie in a common half-plane, with \(q_k\) in the middle. Those which do not lie in a common half-plane correspond to the bottom face \(F_0\) opposite \(P\). 
		\item[\ding{226}] The interior of the tetrahedron are those configurations where the particles are in general position. 
	\end{itemize}
	
	Permuting the particles generates the \(\mathbf{S}_3\)-action which permutes the coordinates \(x_1, x_2, x_3\) and fixes \(P\). In addition, the \(\mathbf{Z}_2\)-symmetry which negates a particle's position negates a pair of \(x_j\)-coordinates and transposes \(P\) with another vertex. Taken together these generate the order-24 group \(\mathbf{T_d}\) of tetrahedral symmetries.
	
	\begin{rmk}\label{actually}
		We shall actually be considering the quotient by \(\mathbf{SO}(3)\). There is an additional invariant \(\Delta=\langle q_1,q_2\times q_3\rangle\) which satisfies \(\Delta^2=\det Q^TQ\). The shape space is therefore two copies \(\mathcal{T}_\pm\) corresponding to the sign of \(\Delta\), and the union is taken over their boundary for \(\Delta=0\). Topologically the shape space is a 3-sphere. However, since the \(\mathbf{Z}_2\)-symmetry which negates \((q_1, q_2, q_3)\) interchanges \(\mathcal{T}_+\) with \(\mathcal{T}_-\), it will suffice to deal with a single copy of \(\mathcal{T}\).
	\end{rmk}
	\subsection{A Web of Cayley Cubics on the 3-Sphere}
	We shall now suppose that the three particles each have unit mass. In the same way that we derived Eq.~\eqref{nboday_inertia} the inertia tensor is
	\begin{equation}
		\inertia_Q=3\Id-QQ^T.
	\end{equation}
	\begin{propn}
		The 3-web on \(\mathcal{T}\) is the family of Cayley cubics
		\begin{equation}\label{cayley_cubics}
			(\lambda-2)^3+2x_1x_2x_3-(\lambda-2)(x_1^2+x_2^2+x_3^2)=0
		\end{equation}
		for \(\lambda\in[0,3]\) a given eigenvalue of \(\inertia_Q\). The inertia tensor has twice-repeated eigenvalues along the four lines which originate at a vertex and go through the midpoint of the opposite face. There is a triple-repeated eigenvalue at the centre \(x_1=x_2=x_3=0\) of the tetrahedron where these lines intersect.
	\end{propn}
	\begin{proof}
		The characteristic polynomials of \(\pi(Q)\) and \(\inertia_{Q}\) coincide, and thus, the subset of \(\mathcal{T}\) for a constant eigenvalue \(\lambda\) are the surfaces \(\chi(\lambda)=0\), where \(\chi(\lambda)\) is the characteristic polynomial of \(\pi(Q)\) and is given by the left hand side of  Eq.~\eqref{cayley_cubics}.
		
		A repeated eigenvalue occurs whenever the discriminant
		\[
		(x_1^2+x_2^2+x_3^2)^3-27x_1^2x_2^2x_3^3
		\]
		of \(\chi\) is zero. By the AM-GM inequality this only holds when \(x_1^2=x_2^2=x_3^2\). 
	\end{proof}
	The surface defined by taking the equality in Eq.~\eqref{curvytetrahedron_defn} is Cayley's Nodal Cubic Surface. It consists of a curvy tetrahedron \(\partial\mathcal{T}\) with nodal singularities at the vertices, around which  4 conical regions emanate outwards. For \(\lambda\ne 2\) the surface defined by Eq.~\eqref{cayley_cubics} is a dilation of the Cayley cubic by a factor of \(\lambda-2\). Notice that for \(\lambda<2\) this scaling is negative and inverts the surface, producing a tetrahedron dual to \(\mathcal{T}\). The intersections of these surfaces with \(\mathcal{T}\) produce the leaves of the 3-web and are shown in Figure~\ref{octych}.
	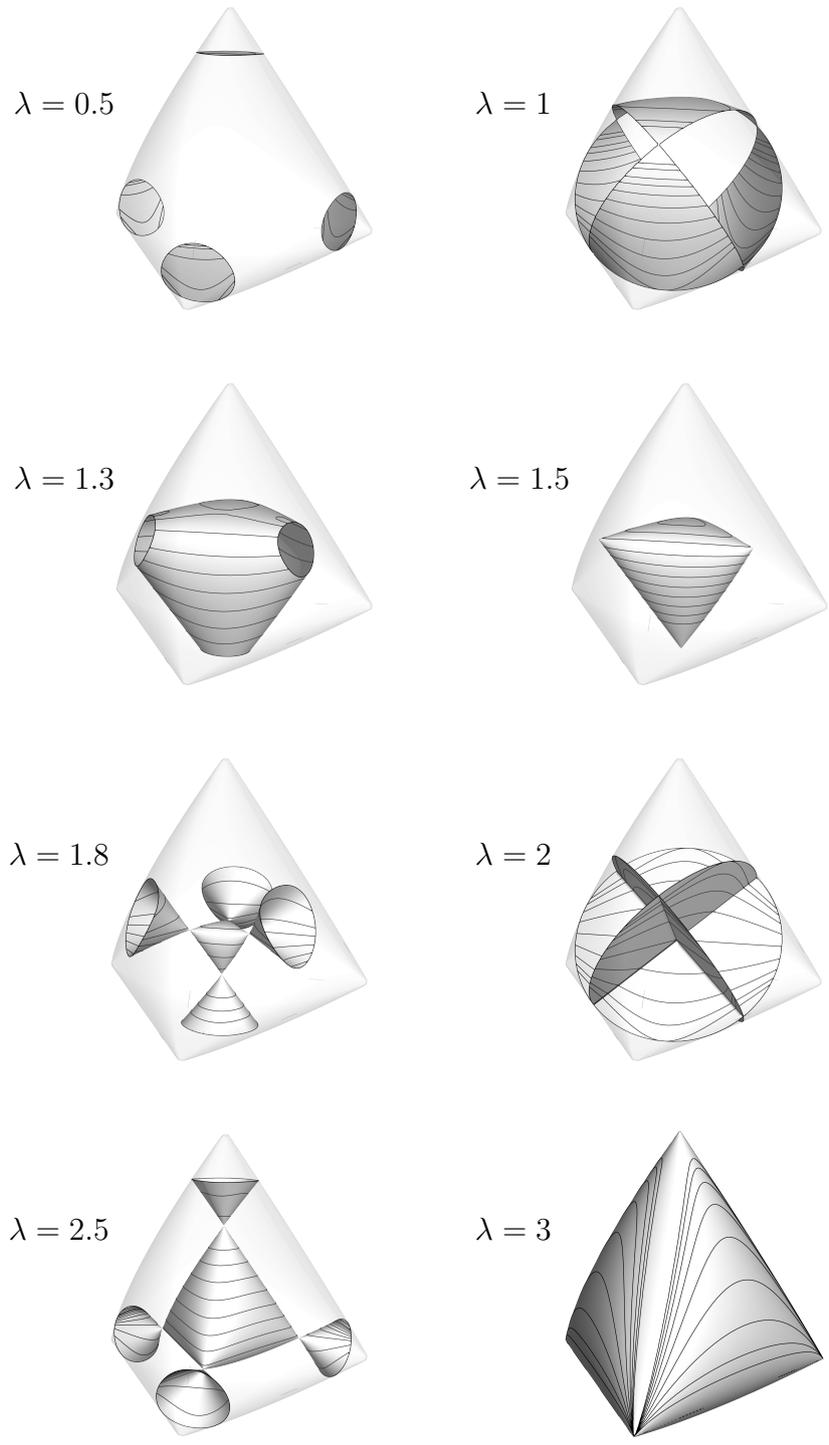
\begin{figure}
		\centering
	\begin{tabular}{cc}
	\begin{tikzpicture}
		\draw (0, 0) node[inner sep=0] {\includegraphics[scale=0.6]{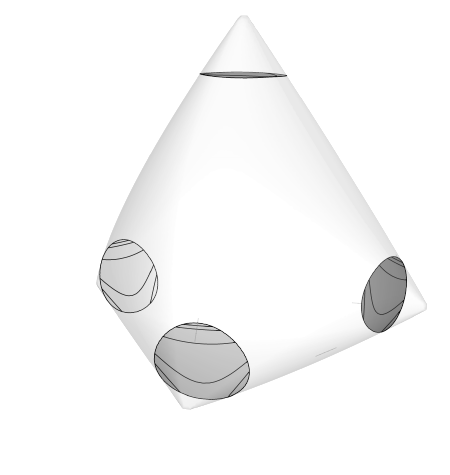}};
		\draw (-2,1) node {$\lambda=0.5$};
	\end{tikzpicture}\quad\quad& 
	\begin{tikzpicture}
		\draw (0, 0) node[inner sep=0] {\includegraphics[scale=0.6]{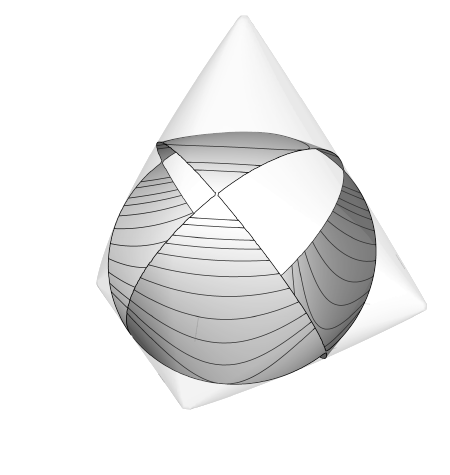}};
		\draw (-2,1) node {$\lambda=1$};
	\end{tikzpicture}
	
	\\
	
	\begin{tikzpicture}
		\draw (0, 0) node[inner sep=0] {\includegraphics[scale=0.6]{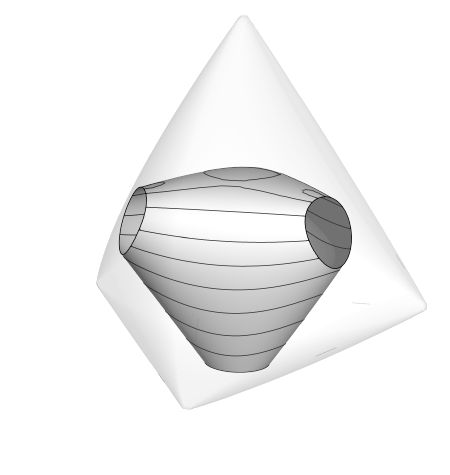}};
		\draw (-2,1) node {$\lambda=1.3$};
	\end{tikzpicture}\quad\quad&
	\begin{tikzpicture}
		\draw (0, 0) node[inner sep=0] {\includegraphics[scale=0.6]{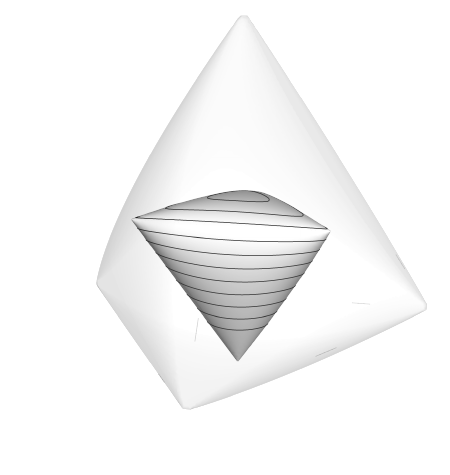}};
		\draw (-2,1) node {$\lambda=1.5$};
	\end{tikzpicture}
	
	\\

	\begin{tikzpicture}
		\draw (0, 0) node[inner sep=0] {\includegraphics[scale=0.6]{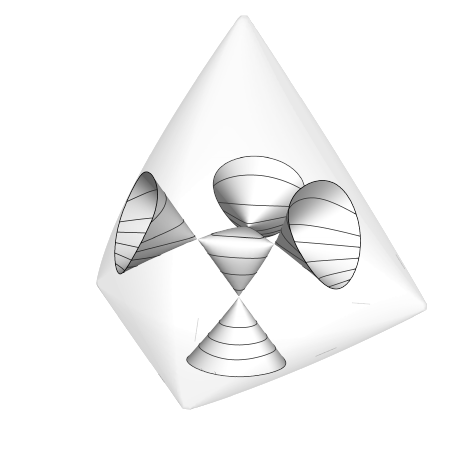}};
		\draw (-2,1) node {$\lambda=1.8$};
	\end{tikzpicture} \quad\quad& 
	\begin{tikzpicture}
		\draw (0, 0) node[inner sep=0] {\includegraphics[scale=0.6]{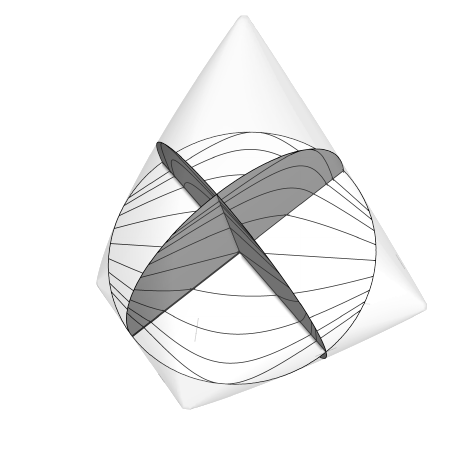}};
		\draw (-2,1) node {$\lambda=2$};
	\end{tikzpicture}
	
	\\

	\begin{tikzpicture}
		\draw (0, 0) node[inner sep=0] {\includegraphics[scale=0.6]{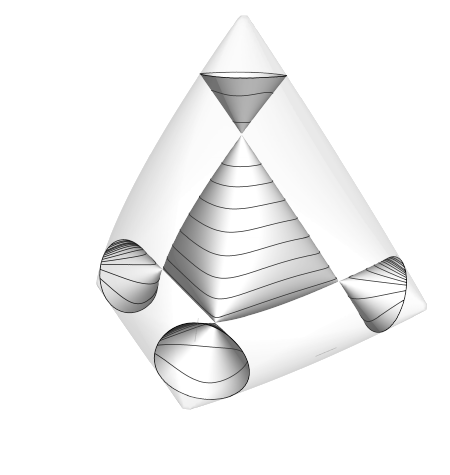}};
		\draw (-2,1) node {$\lambda=2.5$};
	\end{tikzpicture} \quad\quad& 
	\begin{tikzpicture}
		\draw (0, 0) node[inner sep=0] {\includegraphics[scale=0.6]{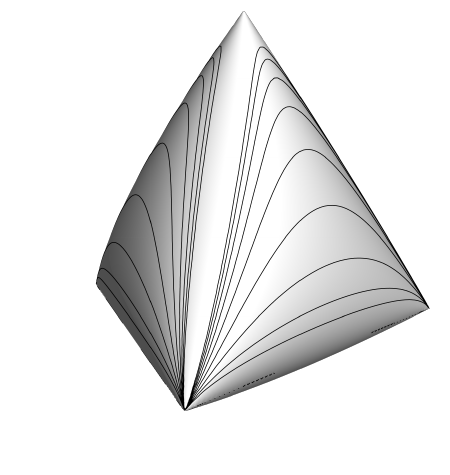}};
		\draw (-2,1) node {$\lambda=3$};
	\end{tikzpicture}
\end{tabular}
		\caption{\label{octych}The 3-web of Cayley cubics for the spherical 3-body problem. The contours are for the cotangent potential.}
	\end{figure}
	\begin{itemize}
		\item[\ding{226}] For \(\lambda=0\) the leaf is singular and consists of the four vertices of \(\mathcal{T}\).  For \(0<\lambda<1\) the leaf becomes 4 disconnected disk-like regions near each vertex.
		\item[\ding{226}] At \(\lambda=1\) the 4 disks connect to each other pairwise at the midpoints of the edges of \(\mathcal{T}\). For \(1\le\lambda\le3/2\) the leaf is a connected surface.
		\item[\ding{226}] For \(\lambda=3/2\) the leaf is a curvy tetrahedron dual to \(\mathcal{T}\) and contained entirely inside \(\mathcal{T}\), with its vertices at the midpoints of the faces of \(\partial\mathcal{T}\). As \(\lambda\) grows the dual tetrahedron shrinks inside \(\mathcal{T}\) and its 4 conical regions grow.
		\item[\ding{226}] The leaf degenerates at \(\lambda=2\) into the union of planes \(x_1,x_2,x_3=0\). For \(2<\lambda\le3\) the leaf is a dilation of \(\partial\mathcal{T}\) and continues to grow until \(\lambda=3\), at which point the leaf coincides with the boundary of \(\mathcal{T}\). For \(\lambda>3\) the intersection is empty and the leaves are no longer defined.
	\end{itemize} 
	\subsection{Classification of Relative Equilibria for Equal Masses}
	\begin{thm}\label{classification_thm}
		Consider 3 particles of equal mass constrained to a sphere and mutually interacting via a strictly attractive potential force depending only on the angle between particles. The system admits the following types of relative equilibrium solutions.
		\begin{enumerate}
			\item \textup{\textbf{Coplanar:}} (i) with one particle on the axis of rotation and the other two particles each located an angle \(\theta\in(0,\frac{\pi}{2})\cup(\frac{\pi}{2},\frac{2\pi}{3})\) either side of the axis; (ii) with one particle orthogonal to the axis of rotation and the other two particles located at an angle \(\theta\in(\frac{2\pi}{3},\pi)\) either side of the first particle; (iii) an equilibrium solution with the particles at the vertices of an equatorial equilateral triangle, and a family of relative equilibria obtained by spinning this configuration in the plane.
			\item \textup{\textbf{General Position:}} (i) with the particles at the vertices of a spherical equilateral triangle with internal angles \(\phi\in(0,\frac{2\pi}{3})\), rotating about the axis through the midpoint of the triangle.
		\end{enumerate}
		These relative equilibria are normal expect for abnormal solutions of Type 1.(i) for \(\theta=\frac{\pi}{3}\), and Type 2.(i) for \(\phi=\frac{\pi}{2}\). For the specific choice of \textup{\textbf{cotangent potential}}
		\begin{equation}
			V=-\cot\theta_{12}-\cot\theta_{13}-\cot\theta_{23}
		\end{equation}
		the relative equilibria are completely classified by two additional subtypes of normal relative equilibria.
		\begin{enumerate}
			\item (iv)~\textup{\textbf{Scalene:}} with the particles belonging to a half plane containing the axis of rotation, where the angles \(\alpha\) and \(\beta\) from the middle particle to the other two satisfy
			\begin{equation}\label{scalene_horrible}
				\sin2\beta(\csc^2\alpha+\csc^2\gamma)+\sin2\gamma(\csc^2\alpha-\csc^2\beta)=\sin2\alpha(\csc^2\beta+\csc^2\gamma),
			\end{equation}
			for \(\gamma=\alpha+\beta\) and \(\alpha\ne\beta\). 
			\item[2.](ii)~\textup{\textbf{Isosceles:}} with two particles separated by an angle \(\beta\) and the third particle located at an angle \(\alpha\ne\beta\) from each of the two particles, where \(\alpha,\beta\in(-\pi,\pi)\) satisfy
			\begin{equation}\label{isosceles_eqn}
				\cos\alpha(2\sin^6\alpha-\sin^6\beta)=\sin^3\alpha\sin^3\beta\cos\beta.
			\end{equation}
		\end{enumerate}
	\end{thm}
	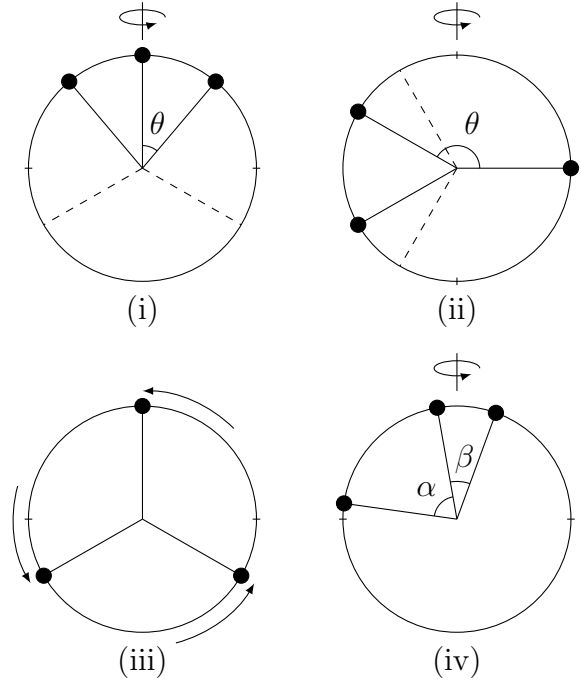
\begin{figure}
		\centering
	\begin{tabular}{cc}
	\begin{tikzpicture}[scale=1]
		\path[use as bounding box] (-2,-2) rectangle (2,2.5);	
		
		\draw (0,0) circle [radius=1.5];
		\draw (0,0) -- (90:1.5);
		\draw[dashed] (0,0) -- (-30:1.5);
		\draw[dashed] (0,0) -- (210:1.5);
		\draw (0,0) -- (50:1.5);
		\draw (0,0) -- (130:1.5);
		
		\draw[fill] (0,1.5) circle [radius=0.1];
		\draw[fill] (50:1.5) circle [radius=0.1];
		\draw[fill] (130:1.5) circle [radius=0.1];
		
		\draw (0,1.7) -- (0,2.2);
		\draw[-latex] (0.3,2) arc (0:310:0.3 and 0.1);
		
		\draw (1.45,0) -- (1.55,0);
		\draw (-1.45,0) -- (-1.55,0);
		
		\draw (50:0.3) arc(50:90:0.3);
		\node at (0.2,0.6){$\theta$};
		\node at (0,-1.9){(i)};
	\end{tikzpicture}
	
	\begin{tikzpicture}[scale=1]
		\path[use as bounding box] (-2,-2) rectangle (2,2.5);	
		
		\draw (0,0) circle [radius=1.5];
		\draw (0,0) -- (0:1.5);
		\draw[dashed] (0,0) -- (120:1.5);
		\draw[dashed] (0,0) -- (-120:1.5);
		\draw (0,0) -- (150:1.5);
		\draw (0,0) -- (-150:1.5);
		
		\draw[fill] (1.5,0) circle [radius=0.1];
		\draw[fill] (150:1.5) circle [radius=0.1];
		\draw[fill] (-150:1.5) circle [radius=0.1];
		
		\draw (0,1.7) -- (0,2.2);
		\draw[-latex] (0.3,2) arc (0:310:0.3 and 0.1);
		
		\draw (0,1.45) -- (0,1.55);
		\draw (0,-1.45) -- (0,-1.55);
		
		\draw (0.3,0) arc(0:150:0.3);
		\node at (0.2,0.6){$\theta$};
		
		\node at (0,-1.9){(ii)};
	\end{tikzpicture}
	
	\\
	
	\begin{tikzpicture}[scale=1]
		\path[use as bounding box] (-2,-2) rectangle (2,2.5);	
		
		\draw (0,0) circle [radius=1.5];
		\draw (0,0) -- (90:1.5);
		\draw (0,0) -- (-30:1.5);
		\draw (0,0) -- (210:1.5);
		
		\draw[fill] (0,1.5) circle [radius=0.1];
		\draw[fill] (-30:1.5) circle [radius=0.1];
		\draw[fill] (210:1.5) circle [radius=0.1];
		
		\draw[-latex] (45:1.7) arc (45:90:1.7);
		\draw[-latex] (165:1.7) arc (165:210:1.7);
		\draw[-latex] (-75:1.7) arc (-75:-30:1.7);
		
		\draw (1.45,0) -- (1.55,0);
		\draw (-1.45,0) -- (-1.55,0);
		
		\node at (0,-1.9){(iii)};
	\end{tikzpicture}
	
	\begin{tikzpicture}[scale=1]
		\path[use as bounding box] (-2,-2) rectangle (2,2.5);	
		
		\draw (0,0) circle [radius=1.5];
		
		\draw[fill] (100:1.5) circle [radius=0.1];
		\draw[fill] (70:1.5) circle [radius=0.1];
		\draw[fill] (172:1.5) circle [radius=0.1];
		
		\draw (0,1.7) -- (0,2.2);
		\draw[-latex] (0.3,2) arc (0:310:0.3 and 0.1);
		
		\draw (1.45,0) -- (1.55,0);
		\draw (-1.45,0) -- (-1.55,0);
		
		\draw (0,0) -- (100:1.5);
		\draw (0,0) -- (70:1.5);
		\draw (0,0) -- (172:1.5);
		
		\draw (70:0.5) arc (70:100:0.5);
		\draw (100:0.3) arc (100:172:0.3);
		
		\node at (0.1,0.8){$\beta$};
		\node at (-0.4,0.4){$\alpha$};
		
		\node at (0,-1.9){(iv)};
	\end{tikzpicture}
\end{tabular}
		\caption{\label{coplanar_pic}Coplanar relative equilibria.}
	\end{figure}
	We shall prove this theorem by classifying the critical points \(\crit(V,\lambda)\). Of course, points in \(\crit(V,\lambda)\) are not necessarily abnormal RE, as one must check that the gradients of \(V\) and \(\lambda\) have the same direction. We leave this routine task as an exercise to the interested reader. 
	
	\subsubsection{Coplanar Configurations}
	The reduced potential is invariant under the \(\mathbf{Z}_2\)-symmetry which interchanges \(\mathcal{T}_+\) with \(\mathcal{T}_-\). Therefore, the points of \(\crit({V},\lambda)\) belonging to \(\partial\mathcal{T}\) are equivalently the points of \(\crit({V}|_{\partial\mathcal{T}},\lambda|_{\partial\mathcal{T}})\). It is also invariant under the \(\mathbf{S}_3\)-symmetry fixing \(P\), including the reflections which transpose two vertices of \(F_0\). The set of critical points therefore contains the intersections of \(\partial\mathcal{T}\) with the planes of symmetry \(x_i=x_j\) for any choice of potential.
	
	\begin{proof}[Proof of Theorem~\ref{classification_thm} for normal RE of Type 1]
		Observe from Eq.~\eqref{cayley_cubics} that the intersection of any leaf from the web with \(\partial\mathcal{T}\) is the same as the intersection of \(\partial\mathcal{T}\) with a sphere centred at the origin. Therefore, for the cotangent potential we must classify the critical points of 
		\[
		V=-\cot\theta_{12}-\cot\theta_{13}-\cot\theta_{23},\quad\text{and}\quad r=\cos^2\theta_{12}+\cos^2\theta_{13}+\cos^2\theta_{23}
		\]
		restricted to \(\partial\mathcal{T}\). For configurations in \(F_0\) we have \(\theta_{12}+\theta_{13}+\theta_{23}=2\pi\). This allows us to write \(V\) and \(r\) in terms of \((\theta_{12},\theta_{23})\). The equation \(\nabla V\times\nabla r=0\) boils down to
		\begin{equation}
			\sin(\theta_{12}-\theta_{13})\sin(\theta_{12}-\theta_{23})\sin(\theta_{13}-\theta_{23})=0
		\end{equation}
		which gives the coplanar configurations of Type~1.{(ii)}. For coplanar configurations in \(F_2\) we now have \(\theta_{13}=\theta_{12}+\theta_{23}\) (recall that the second particle is in the middle of the half-plane containing all three). Writing \(\nabla V\times\nabla r=0\) in terms of \((\alpha,\beta)=(\theta_{12},\theta_{23})\) yields Eq.~\eqref{scalene_horrible} whose solutions are shown in Figure~\ref{scalene_curve_pic}. The curve \(\alpha=\beta\) gives the RE of Type~1.{(i)}, and the other curve gives the scalene family of Type~1.{(iv)}
		
		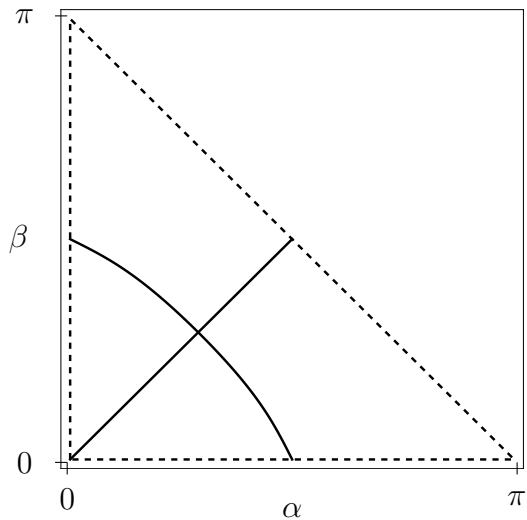
\begin{figure}
			\centering
	\begin{tikzpicture}
	\draw (0,0) node[inner sep=0] {\includegraphics[scale=0.8]{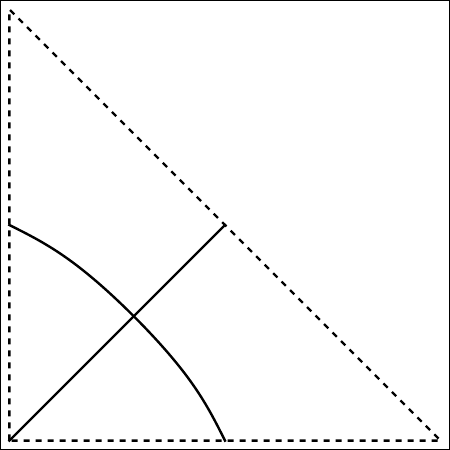}};
	
	\draw (0,-0.8*4.5) node {$\alpha$};
	\draw (-0.8*4.5,0) node {$\beta$};
	
	\draw  (-0.8*3.7,-0.8*3.7) --++(0,-0.8*0.2) ;
	\draw  (0.8*3.7,-0.8*3.7) --++(0,-0.8*0.2) ;
	
	\draw  (-0.8*3.7,0.8*3.7) --++(-0.8*0.2,0) ;

	\draw  (-0.8*3.7,-0.8*3.7) --++(-0.8*0.2,0) ;
	
	\draw (-0.8*3.7,-0.8*4.3) node {$0$};
	\draw (0.8*3.7,-0.8*4.3) node {$\pi$};
	
	\draw (-0.8*4.4,0.8*3.7) node {$\pi$};
	\draw (-0.8*4.4,-0.8*3.7) node {$0$};
	
\end{tikzpicture}
			\caption{\label{scalene_curve_pic}Curve of scalene-coplanar configurations.}
		\end{figure}
		

	\end{proof}
	
	\begin{rmk}
		The boundary \(\partial\mathcal{T}\) is itself the leaf \(\lambda=3\), and is the fixed-point set of the \(\mathbf{Z}_2\)-symmetry which interchanges \(\mathcal{T}_+\) with \(\mathcal{T}_-\). It follows that \(\nabla\lambda=0\) for this eigenvalue on the boundary. For this reason, the equilibrium point at the midpoint of \(F_0\) is also a normal RE for any \(L\) an eigenvector with eigenvalue \(\lambda=3\). These are the RE of Type~1.{(iii)}. 
	\end{rmk}

The RE of Type 1 are displayed in Figure~\ref{coplanar_pic}.

	\subsubsection{Abnormal Relative Equilibria}
	For abnormal RE we must find the critical points of \(V_{\omega}\) in Eq.~\eqref{augmented_pot} for when \(\omega\) ranges over the eigenspace for a repeated eigenvalue. 
	
	\begin{propn}\label{abnormal_propn}
		Along the 4 lines in \(\mathcal{T}\) where \(\inertia_{x}\) has a repeated eigenvalue, the map 
		\[
		\omega\longmapsto\nabla K_x(\omega)
		\]
		from the repeated eigenspace into \(T_x\mathcal{T}\) is a 2-1 map with image
		\begin{equation*}
			\begin{cases}
				T_x\mathcal{T}, &\text{for $x$ at the centre of $\mathcal{T}$}\\[1ex]
				T_x\partial\mathcal{T},&\text{for $x$ at the midpoint of $F_0$}\\[1ex]
				\mathcal{C},&\parbox[t]{7cm}{for $x$ elsewhere on the line through $P$.}
			\end{cases}
		\end{equation*}
		Here \(\mathcal{C}\) is the cone of vectors \(X_1(\partial/\partial x_1)+X_2(\partial/\partial x_2)+X_3(\partial/\partial x_3)\) satisfying
		\begin{equation}\label{abnormal_cone}
			X_1X_2+X_1X_3+X_2X_3=0,\quad\text{for}\quad X_1+X_2+X_3\ge 0.
		\end{equation}
		Along any of the other lines we may apply the \(\mathbf{Z}_2\)-symmetry exchanging \(P\) with another vertex by negating a pair of coordinates from \(x_1,x_2,x_3\).
	\end{propn}
	\begin{proof}
		For \(x\) in the interior of \(\mathcal{T}\) the matrix square root \(x\mapsto\sqrt{3\text{Id}-\pi(Q)}\) is a differentiable local section \(\sigma\). For this choice of section \(\inertia_{\sigma(x)}=\pi(Q)\), and so it suffices to compute
		\[
		\nabla\langle\pi(Q)\omega,\omega\rangle
		\]
		for \(\omega\) an eigenvector of \(\pi(Q)\) with a repeated eigenvalue. 
		
		
		The situation on the boundary is a little different since \((x_1,x_2,x_3)\) fails to be a coordinate chart on shape space. Instead, we take a chart \((\theta_{12},\theta_{13},\varphi)\) with a specific section \(q_1=(0,0,1)^T\), \(q_2=(\sin\theta_{12},0,\cos\theta_{12})^T\), and \(q_3=(\cos\varphi\sin\theta_{13},\sin\varphi\sin\theta_{13},\cos\theta_{13})^T\), and compute \(\nabla K_x(\omega)\) directly.
		
	\end{proof}
	
	\begin{proof}[Proof of Theorem~\ref{classification_thm} for abnormal RE]
	It follows from the proposition that exactly one abnormal RE exists (up to time-reversal) at the centre of \(\mathcal{T}\) and at the midpoints of the faces. 
		
		The gradient \(\nabla V\) is proportional to \((1,1,1)\) when evaluated along any of the four lines in the interior of \(\mathcal{T}\). Hence, for \(x\) elsewhere along these lines there are no critical points of \(V_\omega\) since \(\nabla{V}\) does not belong to the cones in the previous proposition.

	\end{proof}
	
	%
	\subsubsection{General Position Configurations}
	\begin{lem}
		For the cotangent potential every point of \(\crit({V},\lambda)\) in the interior of \(\mathcal{T}\) is isosceles. That is to say, \(x_i=x_j\) for a pair of coordinates.
	\end{lem}
	\begin{proof}
		Consider a leaf of the web. We immediately discount the degenerate leaf \(\lambda=2\) since \({V}\) restricted to the coordinate planes \(x_k=0\) is regular everywhere. 
		
		For \(\lambda\ne 2\) the component of the web which is a curvy tetrahedron admits a parametrisation
		\begin{equation}\label{cayley_param}
			x_j=(\lambda-2)\cos\varphi_j,\quad\text{for}\quad\varphi_1+\varphi_2+\varphi_3=0.
		\end{equation}
		Finding critical points of \({V}\) restricted to the leaf therefore becomes the Lagrange multiplier problem
		\[
		(\lambda-2)\sin\varphi_j=\kappa\left[1-(\lambda-2)^2\cos^2\varphi_j\right]^{3/2}.
		\]
		Squaring both sides and rearranging yields
		\[
		\frac{(\lambda-2)^2-x_j^2}{\kappa^2}=(1-x_j^2)^3.
		\]
		Note that this is a cubic in \(x_j^2\). Suppose that the solutions \(x_1,x_2,x_3\) are distinct. Then \(x_1^2,x_2^2,x_3^2\) must be the three solutions to the cubic, and so their sum must be the coefficient of \(x_j^4\), which is \(3\). The only \(x\) in \(\mathcal{T}\) for which \(x_1^2+x_2^2+x_3^2=3\) are the vertices. 
		
		The vector \(\nabla\lambda\) is tangent to \(x_i=-x_j\) as it is a plane of symmetry for \(\mathcal{T}\). There can be no RE in this plane since all of the entries of \(\nabla V\) are strictly negative.
		
		For the conical components of the web we repeat the argument but using the parametrisation \(x_j=\pm(\lambda-2)\cosh\varphi_j\).
	\end{proof}
	
	\begin{proof}[Proof of Theorem~\ref{classification_thm} for RE of Type 2]
		The line through \(P\) and the midpoint of \(F\) is fixed by the \(\mathbf{S}_3\)-symmetry. Therefore, for any choice of potential this line belongs to \(\crit({V},\lambda)\), and for a strictly attractive potential corresponds to a family of normal RE of Type 2.{(i)}.
		
		To classify the remaining points of \(\crit({V},\lambda)\) in the interior of \(\mathcal{T}\) it suffices from the previous lemma to consider the plane \(x_1=x_3\). In this plane Eq.~\eqref{cayley_cubics} factors as
		\begin{equation}\label{webs_inthe_plane}
			\left((\lambda-2)+x_2\right)\left((\lambda-2)^2+(\lambda-2)x_2-2x_1^2\right)
		\end{equation}
	There are no critical points along the line \(x_2=2-\lambda\). If we parametrise the parabola defined by the quadratic factor in Eq.~\eqref{webs_inthe_plane} by \(x_1\) we find that 
	\begin{equation}
		\label{rate_of_change}
		\frac{dV}{dx_1}=-\frac{2}{(1-x_1^2)^{3/2}}-\left(\frac{4x_1}{\lambda-2}\right)\frac{1}{(1-x_2^2)^{3/2}}.
	\end{equation}
We can solve for \(\lambda\) in terms of \(x_1\) and \(x_2\), and then substitute this into the equation above. By setting the resulting expression to zero we find that the critical points are those which satisfy
		\begin{equation}\label{nasty_curve_eqn}
			\pm(1-x_2^2)^{3/2}\sqrt{8x_1^2+x_2^2}=4x_1(1-x_1^2)^{3/2}-x_2(1-x_2^2)^{3/2}.
		\end{equation}
		This defines three curves in the plane shown in Figure~\ref{iso_contour}, including the line \(x_1=x_2=x_3\) of equilateral solutions of Type 2.{(i)}, and the curves of Type 2.{(ii)}. By squaring both sides of this equation and writing \(x_1=\cos\alpha\) and \(x_2=\cos\beta\) we obtain Eq.\eqref{isosceles_eqn}.
	\end{proof}
	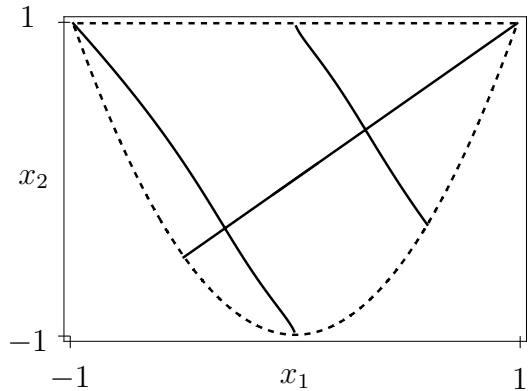
\begin{figure}
		\centering
\begin{tikzpicture}
	\draw (0,0) node[inner sep=0] {\includegraphics[scale=0.8]{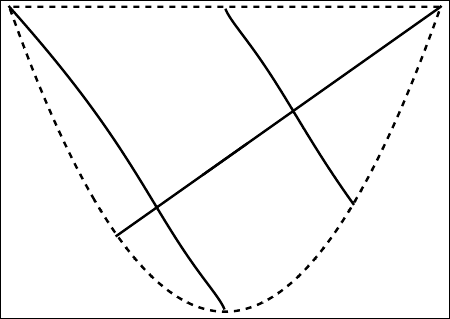}};
	
	\draw (0,-0.8*3.3) node {$x_1$};
	\draw (-0.8*4.3,0) node {$x_2$};
	
	\draw  (-0.8*3.7,-0.8*2.6) --++(0,-0.8*0.2) ;
	\draw  (0.8*3.7,-0.8*2.6) --++(0,-0.8*0.2) ;
	
	\draw  (-0.8*3.7,0.8*2.6) --++(-0.8*0.2,0) ;

	\draw  (-0.8*3.7,-0.8*2.6) --++(-0.8*0.2,0) ;
	
	\draw (-0.8*3.7,-0.8*3.3) node {$-1$};
	\draw (0.8*3.7,-0.8*3.3) node {$1$};
	
	\draw (-0.8*4.4,0.8*2.7) node {$1$};
	\draw (-0.8*4.4,-0.8*2.7) node {$-1$};
	
\end{tikzpicture}
		\caption{\label{iso_contour}The intersection of $\mathcal{T}$ with the plane $x_1=x_3$. Shown are the curves of Type~2 relative equilibria.}
	\end{figure}

\section{Stability}
A RE will be called stable if the corresponding equilibrium in the symplectic reduced space is stable. A sufficient condition for the stability of an equilibrium is that it be a local extremum of the Hamiltonian function; this is known as Dirichlet's criterion. The energy-momentum method seeks to establish the stability of a RE by showing that the corresponding equilibrium in reduced space is a local extremum of the reduced Hamiltonian.

\subsection{The Energy-Momentum Method}

We return to Eq.~\eqref{H_first_deriv} for the rate of change in the Hamiltonian along a curve \(\gamma(t)\) in the reduced space \(T^*U\times\orb\). Differentiating a second time yields
\begin{multline*}
	\frac{d^2}{dt^2}H(\gamma(t))=\langle\dot{y},\M_x^{-1}\dot{y}\rangle+\dots+\left(\langle\ddot{L},\inertia_{x}^{-1}L\rangle+\langle\dot{L},\dot{\inertia}_x^{-1}L\rangle+\langle\dot{L},\inertia_x^{-1}\dot{L}\rangle\right)\\
	+\frac{1}{2}\left(\langle\dot{L},\dot{\inertia}_x^{-1}L\rangle+\langle L,\ddot{\inertia}_x^{-1}L\rangle+\langle L,\dot{\inertia}_x^{-1}\dot{L}\rangle\right)+\ddot{V},
\end{multline*}
where the ellipsis indicates terms which contain \(y\). Indeed, after we set \(t=0\) we have \(y(0)=0\) and the expression above simplifies to
\begin{equation*}
	\langle\dot{y},\M_{x_0}^{-1}\dot{y}\rangle+2\langle\dot{L},\dot{\inertia}_x^{-1}L_0\rangle+\frac{d^2}{dt^2}\bigg|_{t=0}\left(V_{L_0}(x(t))+J_{x_0}(L(t))\right).
\end{equation*}
For normal RE we shall suppose, as in the proof of Theorem~\ref{main_thm}, that the section \(\sigma\) is such that \(L_0\) is a critical point of \(J_{x}\) restricted to \(\orb\) for all \(x\) in \(U\). Therefore, by Lemma~\ref{lemma}, \(\omega_0=\inertia_x^{-1}L_0\) satisfies \(\ad^*_{\omega_0} L_0=0\) for all \(x\). This implies that \(\inertia_x^{-1}L_0\) belongs to the isotropy subalgebra \(\mathfrak{g}_{L_0}\), and therefore, so too does the derivative \(\dot{\inertia}_x^{-1}L_0\). The tangent space to \(\orb\) at \(L_0\) is the annihilator \(\mathfrak{g}_{L_0}^\circ\), and hence, the term \(\langle\dot{L},\dot{\inertia}_x^{-1}L_0\rangle\) in the expression above vanishes, giving us

\begin{propn}
	Let \((x_0,L_0)\) be a normal relative equilibrium. The Hessian of the reduced Hamiltonian evaluated at the fixed point \((x_0,0,L_0)\) in \(T^*U\times\orb\) is given in block-diagonal form by
	\begin{equation}\label{hess}
		\begin{pmatrix}
			\M_{x_0}^{-1} & 0 & 0 \\
			0 &\Hess(V_{L_0}) & 0\\
			0 & 0 & \Hess(J_{x_0})
		\end{pmatrix}.
	\end{equation}
\end{propn}

	Eq.~\eqref{hess} coincides with the block diagonalisation obtained from the reduced energy-momentum method of \cite[Thm.~2.7]{stability}. The rigid `rotational modes' correspond to the so-called Arnold form \(\text{Hess}(J_{x_0})\), and the internal `vibrational modes' to the block \(\text{Hess}(V_{L_0})\).
	

\begin{rmk}\label{signature_remark}
The \(\mathbb{M}_{x_0}^{-1}\) block in Eq.~\eqref{hess} is always positive definite as it is the metric tensor for the reduced metric on shape space. Therefore, the energy-momentum method can only return a positive stability result when \(\text{Hess}(V_{L_0})\) and \(\text{Hess}(J_{x_0})\) are both positive definite. That being said, the signature of the Hessian in  Eq.~\eqref{hess} can be used to deduce instability. If the Hessian is non-singular and has an odd number of positive eigenvalues, then the equilibrium is unstable.
\end{rmk}

\begin{ex}[The rigid body]
	For the case of rotational symmetry the quadratic form \(J_{x_0}\) produces the familiar picture of Euler's equations for the rigid body. The critical points of \(J_{x_0}\) on the coadjoint orbit \(\mathbf{S}^2_\mu\) are the eigenvectors of \(\inertia_{x_0}\). These have signature \((++)\) for the largest eigenvalue, \((+-)\) for the intermediate eigenvalue, and \((--)\) for the smallest eigenvalue.
\end{ex}

\subsection{Stability of Eulerian and Lagrangian Solutions}
\begin{thm}\label{3bp_stab}
	Consider three particles of unit mass on \(\textup{\textbf{S}}^2\) whose  mutual interaction is given by the cotangent potential. For relative equilibria of Types 1.(i)--(iii) and Type 2.(i) the signature of the Hessian at the corresponding equilibria in reduced space is as follows:
	\begin{enumerate}
		\item \textup{\textbf{Eulerian solutions}} of Types 1.(i)--(ii) with separation angle \(\theta\) have signature
		\begin{itemize}
			\item [\ding{226}]\((+++,++-,--)\) for \(0<\theta<\theta_\textup{scal}\),
			\item [\ding{226}]\((+++,+++,--)\) for \(\theta_\textup{scal}<\theta<\theta_\textup{iso}\),
			\item [\ding{226}]\((+++,-++,--)\) for \(\theta_\textup{iso}<\theta<\frac{\pi}{3}\),
			\item [\ding{226}]\((+++,-+-,+-)\) for \(\frac{\pi}{3}<\theta<\frac{2\pi}{3}\), and
			\item [\ding{226}]\((+++,-++,+-)\) for \(\frac{2\pi}{3}<\theta<\pi\).
			\end{itemize}
		Here \(\theta_\textup{scal}\approx 0.906\) satisfies \(32\cos^6\theta-2\cos^2\theta-1=0\) and \(\theta_\textup{iso}\approx 0.934\) satisfies \(\cos^4\theta=1/8\).
		\item \textup{\textbf{Lagrangian solutions}} of Type 2.(i) with internal angle \(\phi\) have signature
		\begin{itemize}
			\item [\ding{226}]\((+++,--+,--)\) for \(0<\phi<\frac{\pi}{2}-\phi_\textup{scal}\),
			\item [\ding{226}]\((+++,+++,--)\) for \(\frac{\pi}{2}-\phi_\textup{scal}<\phi<\frac{\pi}{2}\),
			\item [\ding{226}]\((+++,--+,++)\) for \(\frac{\pi}{2}<\phi<\frac{\pi}{2}+\phi_\textup{scal}\), and
			\item [\ding{226}]\((+++,+++,++)\) for \(\frac{\pi}{2}+\phi_\textup{scal}<\phi<\frac{2\pi}{3}\).
			\end{itemize}
		Here \(\phi_\textup{scal}\) satisfies \(\sin\phi=1/\sqrt{10}\). Additionally, coplanar solutions of Type 1.(iii) rotating in the plane with angular momentum \(L\) have signature
		\begin{itemize}
			\item [\ding{226}]\((+++,++-,++)\) for \(|L|^2<24\sqrt{3}\), and
			\item [\ding{226}]\((+++,+++,++)\) for \(|L|^2>24\sqrt{3}\).
		\end{itemize}
	\end{enumerate}
The signatures are written according to the grouping of blocks \(\mathbb{M}_{x}^{-1}\oplus\textup{Hess}(V_{L})\oplus\textup{Hess}(J_{x})\) in the tangent spaces to \(T^*U\times\textup{\textbf{S}}^2_\mu\).
\end{thm}

Before commencing the proof we assert the following easily verified statements concerning the inertia eigenvalues of the RE. For Types 1.(i)--(ii) \(\lambda\) is the smallest of the three eigenvalues when \(\theta<\pi/3\), otherwise it is intermediate. For Type 1.(iii) \(\lambda=3\) is maximal. For Type 2.(i) \(\lambda\) is the smallest of the three eigenvalues for \(\phi<\pi/2\), otherwise it is the largest. This establishes the signatures of the block \(\text{Hess}(J_x)\). 

It remains to find the signature of each \(\text{Hess}(V_{L})\) to complete the proof. In order to do this we shall use Remark~\ref{actually} to identify shape space with the smooth affine variety
\[
\mathcal{B}=\left\{(x,\Delta)\in\R^3\times\R~|~C(x)=\Delta^2\right\}.
\]
Here \(C(x)\) is the cubic appearing in Eq.~\eqref{curvytetrahedron_defn}. This will allow us to differentiate curves which pass through coplanar configurations where \(\Delta=0\), since here \((x_1,x_2,x_3)\) fails to serve as a coordinate chart on \(\partial\mathcal{T}\). 

\begin{lem}\label{regularize}
	Let \(x(s)\) be a curve in \(\mathcal{T}^\pm\) which intersects \(\partial\mathcal{T}\) transversally at \(s=s_0\) and \(F\) a function on \(\mathcal{B}\) which depends on \(x\) alone. There is a reparametrisation \(\widetilde{x}(s)=x(\sigma(s))\) with \(\sigma(s_0)=s_0\), and a lift of this to a curve \(\Gamma(s)=(\widetilde{x}(s),\Delta(s))\) in \(\mathcal{B}\) with
	\begin{equation}\label{useful}
		\frac{d^2}{ds^2}\bigg|_{s_0}F(\Gamma(s))=\left(\frac{d}{ds}\bigg|_{s_0}C(x(s))\right)\left(\frac{d}{ds}\bigg|_{s_0}F(x(s))\right)
	\end{equation}
\end{lem}
\begin{proof}
	The reparametrisation is required to ensure that \(\Delta(s)=\pm\sqrt{C(\widetilde{x}(s))}\) be differentiable at \(s_0\). This is achieved by a \(\sigma\) satisfying \(\sigma'=\sqrt{2C(x(\sigma))}\), whose existence is guaranteed by Peano's theorem. 
\end{proof}
\begin{proof}[Proof of Theorem~\ref{3bp_stab} for Lagrangian solutions]
	The points in shape space for the Lagrangian RE are along the line of symmetry through \(P\) and the midpoint of the face \(F_0\). We can exploit the tetrahedral symmetry to decompose the Hessian of \(V_L\) into two subspaces: along this line, and orthogonal to it. 
	
	The space orthogonal to the line is tangent to the surface of constant \(\lambda\). Therefore, we equivalently seek the Hessian of \(V\) restricted to this leaf. The parametrisation given in Eq.~\eqref{cayley_param} allows us to compute this directly. Both eigenvalues of the Hessian are found to be positive scalar multiples of 
	\[
	(\lambda-2)(5\lambda^2-20\lambda+18).
	\]
	The root \(\lambda=2\) corresponds to \(\phi=\pi/2\), and the root-pair \(\lambda=2\pm\sqrt{2/5}\) to \(\phi=\cos^{-1}(\pm 1/\sqrt{10})\).
	
	We must now determine the Hessian of \(V_L\) in the direction along the line of symmetry \(x_1=x_2=x_3=\cos\phi\). Along this line \(\lambda=2(1-\cos\phi)\), and thus, we have a one-dimensional function 
	\[
	V_L(\phi)=-3\cot\phi+\frac{|L|^2}{4(1-\cos\phi)}
	\]
whose critical points can be shown to be minima for \(0<\phi<2\pi/3\). For RE of Type 1.(iii) we apply Lemma~\ref{regularize} to \(V_L(\phi)\) evaluated at \(\phi=2\pi/3\).
\end{proof}

The points in shape space for the Eulerian RE belong to the intersection of  the boundary \(\partial\mathcal{T}\) with the plane \(x_i=x_j\). The symmetries of the tetrahedron imply that the Hessian of \(V_L\) decomposes into a boundary component in \(\partial\mathcal{T}\), and a transversal component. We divide the proof into two parts accordingly.

\begin{proof}[Proof of Theorem~\ref{3bp_stab} for Eulerian solutions: transversal component]

	The inertia is invariant with respect to the tetrahedral symmetries. It follows that the relevant surface of constant \(\lambda\) meets \(\partial\mathcal{T}\) orthogonally in the plane \(x_i=x_j\). Therefore, it suffices to compute the second derivative of \(V\) evaluated along such a curve at the point where it meets the boundary.
	
	The relevant curve is a straight line for \(2\pi/3<\theta<\pi\) and a downward parabola for \(\theta<2\pi/3\) given by the linear and quadratic factors of Eq.~\eqref{webs_inthe_plane}. With the aid of Figure~\ref{iso_contour} we see that these curves meet the boundary on the left for \(\pi/2<\theta<\pi\), where the potential is found to be increasing. On the other hand, when these curves meet the boundary on the right for \(\theta<\pi/2\), we see by evaluating Eq.~\eqref{rate_of_change} at \(x_1=\cos\theta\) and \(x_2=\cos 2\theta\), that
	\[
	\frac{dV}{dx_1}=2\csc^32\theta\sec\theta(1-8\cos^4\theta).
	\]
changes sign at \(\theta_\text{iso}\). Combining this with Lemma~\ref{regularize} reveals that, along these transversal curves, the Hessian of \(V_L\) contributes a \((+)\) for \(\theta<\theta_\text{iso}\) and a \((-)\) otherwise.
\end{proof}
	
\begin{proof}[Proof of Theorem~\ref{3bp_stab} for Eulerian solutions: boundary component]
	As in the proof of Theorem~\ref{classification_thm} we shall use \((\theta_{12},\theta_{23})\) as coordinates on each face of \(\partial\mathcal{T}\). For coplanar configurations \(\lambda=3\) is a root of Eq.~\eqref{cayley_cubics}. The other two roots are
	\begin{equation*}\label{explicit_lambda}
		\lambda_{\pm}=\frac{1}{2}\left( 3\pm\sqrt{3+2\cos(2\theta_{12})+2\cos(2\theta_{13})+2\cos(2\theta_{23})}\right).
	\end{equation*}
This gives us an explicit expression for \(V_L\) in terms of \((\theta_{12},\theta_{23})\). Computing the Hessian can be made a little easier by exploiting the symmetry through the line \(\theta_{12}=\theta_{23}\) to deduce that \((1,1)\) and \((1,-1)\) are eigenvectors of \(\text{Hess}(V_L|_{\partial\mathcal{T}})\) with eigenvalues \(E_1\) and \(E_2\), respectively.
The proof is completed by solving for \(|L|^2\) in \(\nabla V_L=0\) and then computing \(E_1\) and \(E_2\) evaluated at \(\theta=\theta_{12}=\theta_{23}\). The results are collected below.
	\[
\renewcommand{\arraystretch}{1.5}
\begin{array}{lll}
0<\theta<\frac{\pi}{2}:  & E_1=4(4+\cos 2\theta)\csc^32\theta&>0\\[0.2cm]
  & E_2=-\displaystyle\frac{4(32\cos^6\theta-2\cos^2\theta-1)}{(1+2\cos 2\theta)\sin^32\theta}&\pbox{20cm}{$<0$ for $\theta<\theta_\text{scal}$\\ $>0$ for $\theta_\text{scal}<\theta<\frac{\pi}{3}$\\$<0$ for $\theta>\frac{\pi}{3}$}\\[0.5cm]
 \frac{\pi}{2}<\theta<\frac{2\pi}{3}:  & E_1=6\csc^42\theta\sin 4\theta&>0\\
  & E_2=-4(2+\cos 2\theta)(1+2\cos 2\theta)\csc^32\theta&<0\\[0.5cm]
  \frac{2\pi}{3}<\theta<\pi:  & E_1=\displaystyle-\frac{6(7+8\cos 2\theta+3\cos 4\theta)}{(2+\cos 2\theta)\sin^32\theta}&>0\\
  & E_2=-4(2+\cos 2\theta)(1+2\cos 2\theta)\csc^32\theta&>0
\end{array}
\]
\end{proof}
\begin{rmk}
	By continuity of the Hamiltonian, the Hessian must become singular at the critical values of \(\theta\) and \(\phi\) in Theorem~\ref{3bp_stab} where the sig§nature changes. In fact, one can show that \(\theta_\text{scal}\) is where the curve of Eulerian solutions meets the curve of scalene solutions in Figure~\ref{scalene_curve_pic}, and \(\theta_\text{iso}\) is where the curve of isosceles solutions meets the boundary in Figure~\ref{iso_contour}. Furthermore, \(\frac{\pi}{2}\pm\phi_\text{iso}\) correspond to the two points where the line of Lagrangian solutions intersects the curves of isosceles solutions.
\end{rmk}

We conclude with some concrete stability results by combining Theorem~\ref{3bp_stab} with Remark~\ref{signature_remark} to obtain

\begin{cor}
	Eulerian solutions with angular separation \(\theta\) are unstable if \(\theta<\theta_\textup{scal}\) or if \(\theta_\textup{iso}<\theta<2\pi/3\). Lagrangian solutions with internal angle \(\phi\) are local minima of the reduced Hamiltonian if  \(\pi/2+\phi_\textup{scal}<\phi<2\pi/3\), and are therefore stable. Solutions of Type 1.(iii) are unstable if \(|L|^2\) is less than \(24\sqrt{3}\) and become gyroscopically stabilised if it is greater.
\end{cor}


\begin{thebibliography}{10}
	
	\bibitem{jamesborisov}
	A.~V. Borisov, L.~C. Garc\'{\i}a-Naranjo, I.~S. Mamaev, and J.~Montaldi.
	\newblock \href{https://doi.org/10.1007/s10569-018-9835-7}{Reduction and
		relative equilibria for the two-body problem on spaces of constant
		curvature}.
	\newblock {\em Celestial Mech. Dynam. Astronom.}, 130(43), 2018.
	
	\bibitem{bitching}
	A.~V. Borisov, I.~S. Mamaev, and I.~A. Bizyaev.
	\newblock \href{https://doi.org/10.1134/S1560354716050075}{The spatial problem
		of 2 bodies on a sphere. {R}eduction and stochasticity}.
	\newblock {\em Regul. Chaotic Dyn.}, 21:556--580, 2016.
	
	\bibitem{dedkind}
	R.~Dedekind.
	\newblock Zusatz zu der vorstehenden {A}bhandlung.
	\newblock {\em J. Reine Angew. Math.}, 58:217--228, 1860.
	
	\bibitem{fasso}
	F.~Francesco and D.~Lewis.
	\newblock \href{https://doi.org/10.1007/PL00004245}{Stability properties of the
		{R}iemann ellipsoids}.
	\newblock {\em Arch. Ration. Mech. Anal.}, 158:259--292, 2001.
	
	\bibitem{fuji1}
	T.~Fujiwara and E.~P{\'{e}}rez-Chavela.
	\newblock \href{ https://doi.org/10.48550/arXiv.2202.10351 }{Three-body
		relative equilibria on $\mathbb{S}^2$ I: Euler configurations}.
	\newblock {\em arXiv:2202.10351}, 2022.
	
	\bibitem{fuji3}
	T.~Fujiwara and E.~P{\'{e}}rez-Chavela.
	\newblock \href{ https://doi.org/10.48550/arXiv.2203.14930}{Equal masses
		{E}ulerian relative equilibria on a rotating meridian of $\mathbb{S}^2$}.
	\newblock {\em arXiv:2203.14930}, 2022.
	
	\bibitem{fuji2}
	T.~Fujiwara and E.~P{\'{e}}rez-Chavela.
	\newblock \href{https://doi.org/10.48550/arXiv.2202.12708}{Three-body relative
		equilibria on $\mathbb{S}^2$ II: Extended {L}agrangian configurations}.
	\newblock {\em arXiv:2202.12708}, 2022.
	
	\bibitem{fuji4}
	T.~Fujiwara and E.~P{\'{e}}rez-Chavela.
	\newblock \href{ https://doi.org/10.48550/arXiv.2304.13782}{A new method to
		study relative equilibria on $\mathbb{S}^2$}.
	\newblock {\em arXiv:2304.13782}, 2023.
	
	\bibitem{fuji5}
	T.~Fujiwara and E.~P{\'{e}}rez-Chavela.
	\newblock \href{ https://doi.org/10.48550/arXiv.2306.13838}{Continuations and
		bifurcations of relative equilibria for the positive curved three body
		problem}.
	\newblock {\em arXiv:2306.13838}, 2023.
	
	\bibitem{fuji_pub}
	T.~Fujiwara and E.~P{\'{e}}rez-Chavela.
	\newblock \href{https://doi.org/10.1134/S1560354723040111}{Three-Body Relative
		Equilibria on $\mathbb{S}^2$}.
	\newblock {\em Regul. Chaotic Dyn.}, 28:680--706, 2023.
	
	\bibitem{marsden_lectures}
	J.~E. Marsden.
	\newblock {\em \href{https://doi.org/10.1017/CBO9780511624001}{Lectures on
			mechanics}}, volume 174 of {\em London Mathematical Society Lecture Note
		Series}.
	\newblock Cambridge University Press, Cambridge, 1992.
	
	\bibitem{jamespersistence}
	J.~Montaldi.
	\newblock \href{https://doi.org/10.1088/0951-7715/10/2/009}{Persistence and
		stability of relative equilibria}.
	\newblock {\em Nonlinearity}, 10:449--466, 1997.
	
	\bibitem{jamesmolecule}
	J.~A. Montaldi and R.~M. Roberts.
	\newblock \href{https://doi.org/10.1007/s003329900064}{Relative equilibria of
		molecules}.
	\newblock {\em J. Nonlinear Sci.}, 9:53--88, 1999.
	
	\bibitem{sr_book}
	R.~Montgomery.
	\newblock {\em \href{https://doi.org/10.1090/surv/091}{A tour of subriemannian
			geometries, their geodesics and applications}}, volume~91 of {\em
		Mathematical Surveys and Monographs}.
	\newblock American Mathematical Society, Providence, RI, 2002.
	
	\bibitem{riemann}
	B.~Riemann.
	\newblock {E}in {B}eitrag zu den {U}ntersuchungen über die {B}ewegung eines
	fl\"{u}ssigen gleichartigen {E}llipsoides.
	\newblock {\em Abh. d. K\"{o}ningl. Gesell. der Wiss. zur G\"{o}ttingen},
	9:3--36, 1860.
	
	\bibitem{stability}
	J.~C. Simo, D.~Lewis, and J.~E. Marsden.
	\newblock \href{https://doi.org/10.1007/BF01881678}{Stability of relative
		equilibria. {I}. {T}he reduced energy-momentum method}.
	\newblock {\em Arch. Rational Mech. Anal.}, 115:15--59, 1991.
	
	\bibitem{smale1}
	S.~Smale.
	\newblock \href{https://doi.org/10.1007/BF01418778}{Topology and mechanics.
		{I}}.
	\newblock {\em Invent. Math.}, 10:305--331, 1970.
	
	\bibitem{smale2}
	S.~Smale.
	\newblock \href{https://doi.org/10.1007/BF01389805}{Topology and mechanics.
		{II}. {T}he planar {$n$}-body problem}.
	\newblock {\em Invent. Math.}, 11:45--64, 1970.
	
	\bibitem{smale_list}
	S.~Smale.
	\newblock \href{https://doi.org/10.1007/BF03025291}{Mathematical problems for
		the next century}.
	\newblock {\em Math. Intelligencer}, 20:7--15, 1998.
	
	\bibitem{krish}
	Li~Sheng Wang, P.~S. Krishnaprasad, and J.~H. Maddocks.
	\newblock \href{https://doi.org/10.1007/BF02426678}{Hamiltonian dynamics of a
		rigid body in a central gravitational field}.
	\newblock {\em Celestial Mech. Dynam. Astronom.}, 50:349--386, 1991.
	
	\bibitem{wang}
	Li~Sheng Wang, J.~H. Maddocks, and P.~S. Krishnaprasad.
	\newblock Steady rigid-body motions in a central gravitational field.
	\newblock {\em J. Astronaut. Sci.}, 40:449--478, 1992.
	
\end{thebibliography}

\end{document}